\RequirePackage[l2tabu, orthodox]{nag}		
\documentclass[10pt]{article}

\usepackage[T1]{fontenc}
\usepackage{lmodern}
\usepackage[frenchmath]{mathastext}							
\usepackage{natbib}       									
\usepackage{amsmath}                        
\numberwithin{equation}{section}
\usepackage{graphicx} 											
\usepackage{enumitem}												
\usepackage{mdwlist}												
\usepackage[dvipsnames]{xcolor}							
\usepackage[plainpages=false, pdfpagelabels]{hyperref} 
	\hypersetup{
		colorlinks   = true,
		citecolor    = RoyalBlue, 
		linkcolor    = RubineRed, 
		urlcolor     = Turquoise
	}
\usepackage{amssymb}                        
\usepackage[mathscr]{eucal}                 
\usepackage{dsfont}													
\usepackage[paperwidth=8.5in,paperheight=11in,top=1.25in, bottom=1.25in, left=1.0in, right=1.0in]{geometry} 
\usepackage{mathtools}                      
\mathtoolsset{showonlyrefs=true}          
\usepackage{amsthm}                         
	\allowdisplaybreaks                       
	\theoremstyle{plain}
	\newtheorem{theorem}{Theorem}
	\numberwithin{theorem}{section}
	\newtheorem{lemma}[theorem]{Lemma}       	
	\newtheorem{proposition}[theorem]{Proposition}
	
	\theoremstyle{definition}
	\newtheorem{definition}[theorem]{Definition}
	
	\newtheorem{remark}[theorem]{Remark}

\newcommand{\<}{\langle}
\renewcommand{\>}{\rangle}
\renewcommand{\(}{\left(}
\renewcommand{\)}{\right)}
\renewcommand{\[}{\left[}
\renewcommand{\]}{\right]}

\newcommand\Eb{\mathds{E}}
\newcommand\Fb{\mathds{F}}
\newcommand\Ib{\mathds{1}}

\newcommand\Pb{\mathds{P}}

\newcommand\Rb{\mathds{R}}
\newcommand\Zb{\mathds{Z}}
\newcommand\Nb{\mathds{N}}

\newcommand\Ac{\mathscr{A}}
\newcommand\Bc{\mathscr{B}}

\newcommand\Fc{\mathscr{F}}

\newcommand\Hc{\mathscr{H}}

\newcommand\Oc{\mathscr{O}}
\newcommand\Pc{\mathscr{P}}

\newcommand\eps{\varepsilon}
\newcommand\om{\omega}
\newcommand\Om{\Omega}
\newcommand\sig{\sigma}

\newcommand\Lam{\Lambda}
\newcommand\gam{\gamma}
\newcommand\Gam{\Gamma}
\newcommand\lam{\lambda}
\newcommand\del{\delta}

\newcommand\Acb{\bar{\Ac}}

\newcommand\fb{\overline{f}}
\newcommand\xb{\bar{x}}
\newcommand\yb{\bar{y}}
\newcommand\zb{\bar{z}}

\newcommand\ub{\bar{u}}

\renewcommand\d{\partial}

\newcommand\ii{\mathtt{i}}
\newcommand\dd{\mathrm{d}}
\newcommand\ee{\mathrm{e}}
\newcommand\mf{\mathfrak{m}}

\newcommand{\half}{\frac{1}{2}}
\newcommand{\thalf}{\tfrac{1}{2}}

\newcommand{\Gamh}{\widehat \Gam}
\begin{document}

\title{Approximate pricing of European and Barrier claims in a local-stochastic volatility setting}

\author{
Weston Barger
\thanks{Department of Applied Mathematics, University of Washington.  \textbf{e-mail}: \url{wdbarger@uw.edu}}
\and
Matthew Lorig
\thanks{Department of Applied Mathematics, University of Washington.  \textbf{e-mail}: \url{mlorig@uw.edu}}
}

\date{This version: \today}

\maketitle

\begin{abstract}
We derive asymptotic expansions for the prices of a variety of European and barrier-style claims in a general local-stochastic volatility setting.  Our method combines Taylor series expansions of the diffusion coefficients with an expansion in the correlation parameter between the underlying asset and volatility process.  Rigorous accuracy results are provided for European-style claims.  For barrier-style claims, we include several numerical examples to illustrate the accuracy and versatility of our approximations.
\end{abstract}

%
%

\section{Introduction}
\label{sec:intro}
Barrier-style claims are among the most liquid path-dependent claims.  As barrier-style claims are generally cheaper than their European counterparts, the former are popular among speculators who wish to bet on market movements while taking advantage of the lower prices barrier-style claims entail.  Yet, despite their widespread use, barrier-style claims remain challenging to price.
\par
In his landmark work, \cite{merton1973theory} was the first to value a down-and-out call in closed form when the underlying stock follows geometric Brownian motion (GBM).  There exists a variety of static hedging results for barrier claims in GBM or GBM-like settings.  For example, in a GBM framework, \cite{bowie-carr-1994} show that the payoff of a down-and-out call with barrier $L$ can be replicated by buying a European call on the same underlying futures price with the same maturity $T$ and strike $K$ and also selling $K/L$ puts with strike $L^2/K$.  \cite{carr1998static} prove that this static hedge works in any model with local volatility, provided that the volatility function is symmetric in the log of the futures price relative to the barrier.  \cite{pcs} make clear that symmetry condition is merely sufficient, but not necessary.  The hedge described above for a down-and-out call works provided that there are no jumps over the barrier and the call and put have the same implied volatility at the first passage time to the barrier, a condition referred to as Put Call Symmetry (PCS), which was introduced to finance by \cite{bates1988} as a way to measure skewness.  More recently, \cite{carr-nadtochiy} show how to statically hedge barrier options for a general class of local volatility models.  And \cite{bscpv} develop semi-static hedges for barrier-style claims on price and volatility.
\par
Unfortunately, the restrictive symmetry conditions described above prohibit static hedging results from being applied when the underlying is described by any of the models that are most frequently used to price European options: CEV \cite{CoxCEV}, Heston \cite{heston1993} and SABR \cite{sabr}.  For these models, a number of closed-form pricing formulas have been developed and the associated hedging strategies are dynamic.  \cite{davydov2001pricing} price barrier-style claims in a CEV setting using eigenfunctions expansions.  Assuming zero correlation between the price and volatility-driving process, it is known that the underlying in a stochastic volatility model can be expressed as a time-changed GBM.  As a result, in the zero correlation setting, barrier options can be priced via Fourier Sine series (for double barrier options) or via Fourier Sine transforms (for single barrier options) so long as the Laplace transform of the time integral of the stochastic variance process is known in closed form.  This has been done in \cite{faulhaber2002analytic} for the Heston model and presumably could be carried out in the SABR model using the results of \cite{sabr-exact}, though, a detailed literature search did not reveal any paper in which the zero-correlation SABR computation has been carried out.
\par
Zero correlation stochastic volatility models induce symmetric implied volatility smiles and are not consistent with empirical evidence from equity markets, where smiles exhibit strong at-the-money skews.  It is therefore important to allow the underlying to be correlated with the  volatility-driving process.  
When correlation is non-zero, closed-form formulas for barrier option prices are not available and perturbation methods are often employed.
\cite{lipton2014pricing}, for example, finds approximate barrier options prices by expanding prices in a small parameter, which is equal to the correlation times the vol-of-vol.  \cite{lorig-fouque-jaimungal-1} price barrier options in a fast mean-reverting volatility setting.
And \cite{lorig-2} values barrier options and other claims for a class of multiscale stochastic volatility models (see \cite{fpss} for a review of these models).  Yet the methods of \cite{lipton2014pricing} and \cite{lorig-fouque-jaimungal-1} cannot be applied in the  CEV or SABR settings, and the results in \cite{lorig-2} require a separation of time scales between the price process and the corresponding fast and slow factors of volatility, which may not be realistic in certain markets.
\par
In this paper we consider a very general class of local-stochastic volatility models which naturally include the CEV, Heston and SABR models.  We find approximate prices of barrier-style claims by expanding the coefficients of infinitesimal generator of the underlying as a Taylor series about a fixed point.  The Taylor series expansion method was initially developed for European-style claims in scalar diffusion setting in \cite{pagliarani2011analytical} and later extended to $d$-dimensional diffusions in \cite{lorig-pagliarani-pascucci-2} and \cite{lorig-pagliarani-pascucci-4}.
\par
A significant mathematical challenge arises when extending the  methods developed in \cite{lorig-pagliarani-pascucci-4} for diffusions in $\Rb^d$ to diffusions in strict subsets of $\Rb^d$.  In particular, in $\Rb^d$, the zeroth order approximate transition density of a diffusion is given by a Gaussian kernel.  The Gaussian kernel is a function of the difference of the forward and backward variables.  This symmetry between forward and backward variables greatly simplifies the computations required to obtain higher order corrections to the transition density.  For a diffusion in a strict subset of $\Rb^d$ however, the zeroth order transition density approximation will no longer be a function of the difference of the forward and backward variables.  As a result, the computations needed to obtain higher order corrections to the transition density are significantly more involved.
\par
The rest of this paper proceeds as follows.  In Section \ref{sec:model}, we introduce a general local-stochastic volatility model and describe the option-pricing problems we wish to solve.  In Section \ref{sec:asymptotics}, we develop an asymptotic expansion for options prices.  This expansion leads to a sequence of nested PDE problems, which we solve explicitly in Section \ref{sec:explicit}. In Section \ref{sec:accuracy}, we establish the asymptotic accuracy of our approximation for European options.  Finally, in Section \ref{sec:numerics}, we provide several numerical illustrations of our pricing approximation for barrier-style claims and compare our results to prices obtained via Monte Carlo simulation.  Some concluding remarks are offered in Section \ref{sec:conclusion}.

%
%

\section{Market model}
\label{sec:model}
We consider a market defined on a complete, filtered probability space $(\Om, \Fc, \Fb=(\Fc_t)_{t \geq 0}, \Pb)$.  Here, the measure $\Pb$ represents the market's chosen pricing measure.  Let $S = (S_t)_{t \geq 0}$ be the value of a risky asset.  We suppose the dynamics of $S$ are given by
\begin{align}
S_t
	&=	f(X_t) , \\
\dd X_t
	&=	\mu(X_t,Y_t) \dd t + \sig(X_t, Y_t) \dd W_t , \label{eq:assetmodel}\\
\dd Y_t
	&=	c(X_t,Y_t) \dd t + g(X_t, Y_t) \dd B_t , \label{eq:CIRmodel}\\
\dd \< W, B \>_t
	&=	\rho\, \dd t .
\end{align}
where the function $f$ must be  positive, strictly increasing and $C^2$.  The processes $W = (W_t)_{t \geq 0}$ and $B = (B_t)_{t \geq 0}$ are driftless $(\Pb,\Fb)$-Brownian motions with constant correlation $\rho \in (-1,1)$.  We assume the dynamics of $(X,Y) = (X_t,Y_t)_{t \geq 0}$ are such that $(X,Y)$ has a unique strong solution, at least up until the first exit time of $X$ of some interval $I \subseteq \Rb$.  
\par
For simplicity, we take the risk-free rate of interest to be zero.  Thus, in order to preclude the possibility of arbitrage, the risky asset $S$ must be a martingale.  As a result, the function $\mu$, which controls the drift of $X$, must satisfy
\begin{align}
\mu
	&=	\frac{-f'' \sig^2}{2 f'}  .
\end{align} 
The condition on $\mu$ can be easily derived by computing $\dd f(X_t)$ and setting the $\dd t$-term  to zero.  Typical choices for $f$ are $f(x) = \ee^x$, in which case $\mu = -\thalf \sig^2$, or $f(x) = x$, in which case $\mu = 0$.
\par
We are interested in computing the price of a barrier-style claim, whose payoff at the maturity date $T$ is given by
\begin{align}
\text{Payoff}:&&
\Ib_{\{\tau > T\}} \varphi(X_T) , &&
\tau
	&=	\inf\{ t \geq 0 : X_t \notin I \} , \label{eq:payoff}
\end{align}
where $I$ is an interval in $\Rb$.  For a single-barrier claim with a barrier $L < X_0$ we have $I = (L,\infty)$.  For a double-barrier claim, we have $I = (L,U)$ where $L < X_0 < U$.  We also allow for the possibility that $I = \Rb$, which corresponds to a European claim on $X$.

\begin{remark}
When $I \neq \Rb$, payoffs of the form \eqref{eq:payoff} are \emph{knock-out} style payoffs.  A \emph{knock-in} style payoff is a payoff of the form
\begin{align}
\text{Payoff}:&&
\Ib_{\{\tau \leq T\}} \varphi(X_T) . \label{eq:knock-in}
\end{align}
It is known that the value of a knock-in claim with payoff \eqref{eq:knock-in} is equal to the value of a European claim with payoff $\varphi(X_T)$ minus the value of a knock-out claim with payoff \eqref{eq:payoff}.  Thus, by pricing both knock-out and European style claims we can also price knock-in style claims.
\end{remark}

The value $V_t$ of the claim with payoff \eqref{eq:payoff} at time $t \leq T$ is given by
\begin{align}
V_t
	&=	\Ib_{\{\tau > t\}} u(t,X_t,Y_t) , & 
u(t,x,y)
	&:=	\Eb \( \Ib_{\{\tau > T\}} \varphi(X_T) | X_t = x \in I , Y_t = y\) , \label{eq:V}
\end{align}
Under mild conditions, the function $u$, defined in \eqref{eq:V}, is the unique classical solution of the Kolmogorov Backward equation
\begin{align}
0
	&=	(\d_t + \Ac) u , &
u(T,\cdot)
	&=	\varphi , \label{eq:u.pde}
\end{align}
where $\Ac$, the generator of $(X,Y)$, is given explicitly by
\begin{align}
\Ac
	&=	\mu(x,y) \d_x + \thalf \sig^2(x,y) \d_x^2 + c(x,y) \d_y + \thalf g^2(x,y) \d_y^2 + \rho \sig(x,y) g(x,y) \d_x \d_y ,
	\label{eq:A}
\end{align}
and is defined to act on functions that are twice differentiable and satisfy certain boundary conditions
\begin{align}
\text{dom}(\Ac)
	&:=	\{ g \in C^2 : \lim_{x \to \d I} g(x,y) = 0 \} .
\end{align}
Here we use the notation $\d I$ to indicate a \emph{finite} endpoint of $I$.  So, for example, if $I = (L,\infty)$, then $\Ac$ acts on functions that satisfy $\lim_{x \searrow L} g(x,y) = 0$.  Throughout this paper, we assume a unique classical solution to \eqref{eq:u.pde} exists.  Our goal is to find the solution $u$ of PDE \eqref{eq:u.pde}.  As no explicit solution of \eqref{eq:u.pde} exists for general coefficients $(\mu,\sig,c,g)$, we shall seek instead an explicit approximation for $u$.  

%
%

\section{Formal asymptotic expansion}
\label{sec:asymptotics}
In this section, we will present a formal asymptotic expansion for $u$. To begin, let us introduce some notation.  For any coefficient of $\Ac$ we define
\begin{align}
\chi^\eps(x,y)
	&:=	\chi( \xb + \eps(x - \xb), \yb + \eps(y-\yb) ) , &
\chi
	&\in \{ \mu, \thalf \sig^2, c, \thalf g^2, \sig g\} ,
\end{align}
where $(\xb,\yb)$ is a fixed point and $\eps \in [0,1]$.  Next, we introduce an operator $\Ac^{\eps,\rho}$, which is given explicitly by
\begin{align}
\Ac^{\eps,\rho}
	&=	\mu^\eps(x ,y) \d_x + (\thalf \sig^2)^\eps(x,y) \d_x^2 + c^\eps(x,y) \d_y + (\thalf g^2)^\eps(x,y) \d_y^2 + \rho (\sig g)^\eps(x,y) \d_x \d_y ,
\end{align}
where $\text{dom}(\Ac^{\eps,\rho}) := \text{dom}(\Ac)$. Consider, now, a family of PDE problems, indexed by $(\eps,\rho)$
\begin{align}
0
	&=	(\d_t + \Ac^{\eps,\rho}) u^{\eps,\rho} , &
u^{\eps,\rho}(T,\cdot,\cdot)
	&=	\varphi . \label{eq:u.eps.pde}
\end{align}
Noting that $\Ac^{\eps,\rho}|_{\eps=1}=\Ac$ it follows from \eqref{eq:u.pde} and \eqref{eq:u.eps.pde} that $u^{\eps,\rho}|_{\eps=1}=u$.
Thus, rather than seek an approximation solution to PDE problem \eqref{eq:u.pde} directly, we shall instead seek an approximation solution to PDE problem \eqref{eq:u.eps.pde} by expanding $u^{\eps,\rho}$ in powers of $\eps$ and $\rho$ as follows
\begin{align}
u^{\eps,\rho}
	&=	\sum_{i=0}^\infty \sum_{j=0}^\infty \eps^i \rho^j u_{i,j} , \label{eq:u.expand}
\end{align}
where the functions $(u_{i,j})$ are (at present) unknown.  Once we obtain an approximation for $u^{\eps,\rho}$, our approximation for $u$ will be obtained by setting $\eps = 1$.
\par
Assume for the moment that the coefficients in $\Ac$ are analytic.  We shall see later that the approximation we obtain for $u$ does not require this assumption.  However, making this assumption simplifies the presentation considerably so we will temporarily proceed with it.  As the coefficients of $\Ac$ are analytic, we have
\begin{align}
\Ac^{\eps,\rho}
	&=	\sum_{n=0}^\infty \eps^n \( \Ac_{n,0} + \rho \Ac_{n,1} \) , \label{eq:A.expand} \\
\Ac_{n,0}
	&=	\mu_n(x ,y) \d_x + (\thalf \sig^2)_n(x,y) \d_x^2 + c_n(x,y) \d_y + (\thalf g^2)_n(x,y) \d_y^2 , \label{eq:An0} \\
\Ac_{n,1}
	&=	(\sig g)_n(x,y) \d_x \d_y , \label{eq:An1}
\end{align}
where $\text{dom}(\Ac_{0,0}):=\text{dom}(\Ac)$ and have introduced the notation
\begin{align}
\chi_n(x,y)
	&:=	\frac{1}{n!}\frac{\dd^n}{\dd \eps^n} \chi^\eps(x,y) \Big|_{\eps=0} 
	 =	\sum_{i=0}^n \frac{\d_{\xb}^i \d_{\yb}^{n-i} \chi(\xb,\yb)}{i!(n-i)!} (x-\xb)^i (y-\yb)^{n-i} , \label{eq:gn}
\end{align}
for $\chi \in \{ \mu, \thalf \sig^2, c, \thalf g^2, \sig g \}$.  Observe that $\chi_n$ is the $n$th order term in the Taylor series expansion of $\chi$ about the point $(\xb,\yb)$.
\par
Inserting expansions \eqref{eq:u.expand} and \eqref{eq:A.expand} into PDE problem \eqref{eq:u.eps.pde} and collecting terms with like powers of $\eps$ and $\rho$, we obtain 
\begin{align}
\Oc(1):&&
0
	&=	(\d_t + \Ac_{0,0}) u_{0,0}  , &
u_{0,0}(T,\cdot,\cdot) 
	&=	\varphi , \label{eq:u00.pde} \\
\Oc(\eps^n \rho^k):&&
0
	&=	(\d_t + \Ac_{0,0}) u_{n,k} + \sum_{i=0}^n \sum_{j=0}^k (1 - \delta_{i+j,0})\Ac_{i,j} u_{n-i,k-j}, &
u_{n,k}(T,\cdot,\cdot)
	&=	0 . \label{eq:unk.pde} 
\end{align}
For clarity, we present the lowest order terms explicitly here
\begin{align}
\Oc(\eps):&&
0
	&=	(\d_t + \Ac_{0,0}) u_{1,0} + \Ac_{1,0} u_{0,0}  , &
u_{1,0}(T,\cdot,\cdot)
	&=	0 , \label{eq:u10.pde} \\
\Oc(\rho):&&
0
	&=	(\d_t + \Ac_{0,0}) u_{0,1} + \Ac_{0,1} u_{0,0}  , &
u_{0,1}(T,\cdot,\cdot)
	&=	0 , \label{eq:u01.pde} \\
\Oc(\eps^2):&&
0
	&=	(\d_t + \Ac_{0,0}) u_{2,0} + \Ac_{2,0} u_{0,0} + \Ac_{1,0} u_{1,0}  , &
u_{2,0}(T,\cdot,\cdot)
	&=	0 , \label{eq:u20.pde} \\
\Oc(\eps \rho):&&
0
	&=	(\d_t + \Ac_{0,0}) u_{1,1} + \Ac_{1,1} u_{0,0} + \Ac_{1,0} u_{0,1} + \Ac_{0,1} u_{1,0} , &
u_{1,1}(T,\cdot,\cdot)
	&=	0 , \label{eq:u11.pde} \\
\Oc(\rho^2):&&
0
	&=	(\d_t + \Ac_{0,0}) u_{0,2} + \Ac_{0,1} u_{0,1} + \Ac_{0,2}u_{0,0}, &
u_{0,2}(T,\cdot,\cdot)
	&=	0 . \label{eq:u02.pde} 
\end{align}
The above computation motivates the following definition.

\begin{definition}
\label{def:ubar}
Let $u$ be the unique classical solution of PDE problem \eqref{eq:u.pde}.  We define $\ub_N^\rho$, the \emph{$N$th order approximation of $u$}, as
\begin{align}
\ub_N^\rho(t,x,y)
	&:=	\sum_{i=0}^N \sum_{j=0}^i \eps^j \rho^{i-j} u_{j,i-j}(t,x,y) \Big|_{(\xb,\yb,\eps)=(x,y,1)} , \label{eq:ubar}
\end{align} 
where $u_{0,0}$ satisfies \eqref{eq:u00.pde} and $u_{n,k}$ satisfies \eqref{eq:unk.pde} for $(n,k) \neq (0,0)$.
\end{definition}

\begin{remark}
Observe that we have set $\eps = 1$ in \eqref{eq:ubar} and, as a result, the parameter $\eps$ plays no role in the definition for $\ub_N^\rho$.  Indeed, $\eps$ was introduced merely as an accounting tool in the formal asymptotic expansion above.  As $\eps$ does not appear in the original PDE problem \eqref{eq:u.pde} it should not appear in the approximation for $u$.
\end{remark}

\begin{remark}
Note that we have set $(\xb,\yb)=(x,y)$ in \eqref{eq:ubar}.  This is often a point of confusion, and we wish to clarify how this should be handled.  First, one should solve the sequence of nested PDE problems \eqref{eq:u00.pde}--\eqref{eq:unk.pde} with $(\xb,\yb)$ \emph{fixed}.  To be explicit, let us denote the solution of the $\Oc(\eps^n \rho^k)$ PDE as $u_{n,k}^{\xb,\yb}$.  If one is then interested in the approximate value of $u$ at the point $(x,y)$, one should then compute $u_{n,k}^{\xb,\yb}(x,y)|_{(\xb,\yb)=(x,y)}$ in the sum \eqref{eq:ubar}.  The reason for choosing $(\xb,\yb)=(x,y)$ is as follows.  The small-time behavior of a diffusion is predominantly determined by the geometry of the diffusion coefficients near the starting point of the diffusion $(x,y)$.  In turn, the most accurate Taylor series expansion of any function near the point $(x,y)$ is the Taylor series expansion centered at $(\xb,\yb)=(x,y)$.
\end{remark}

\begin{remark}
As previously mentioned, analyticity of the coefficients of $\Ac$ is \emph{not} required.  Indeed, to construct the $N$th order approximation $\ub_N^\rho$ one requires only the operators $\Ac_{k,j}$ for $k \leq N$.  Thus, the $N$th order approximation of $u$ requires only that the coefficients of $\Ac$ be $C^N$.
\end{remark}

%
%

\section{Explicit expressions}
\label{sec:explicit}
In this section, we provide explicit expressions for the functions $(u_{n,k})$ required to compute $\ub_N^\rho$, the $N$th order approximation of $u$.  We begin with a review of Duhamel's principle.  Let $\Gam_{0,0}$ be the fundamental solution of parabolic operator $(\d_t + \Ac_{0,0})$.  That is,
\begin{align}
0
	&=	(\d_t + \Ac_{0,0})\Gam_{0,0}(\cdot,\cdot,\cdot;T,\xi,\eta) , &
\Gam_{0,0}(T,\cdot,\cdot;T,\xi,\eta)
	&=	\del_{\xi,\eta} . \label{eq:Gam00.pde}
\end{align}
Duhamel's principle states that the the unique classical solution to 
\begin{align}
0
	&=	(\d_t + \Ac_{0,0})u + f , &
u(T,\cdot,\cdot)
	&=	h ,
\end{align}
is given by
\begin{align}
u(t,x,y)
	&=	\Pc_{0,0}(t,T)h(x,y) + \int_t^T \dd s \, \Pc_{0,0}(t,s)f(s,x,y) , 
\end{align}
where we have introduced $\Pc_{0,0}$ the \emph{semigroup} generated by $\Ac_{0,0}$, which is defined as follows
\begin{align}
\Pc_{0,0}(t,s) \varphi(x,y)
	&=	\int_I \dd \xi \int_\Rb \dd \eta \, \Gam_{0,0}(t,x,y;s,\xi,\eta) \varphi(\xi,\eta) , \label{eq:P00}
\end{align}
where $0 \leq t \leq s \leq T$.

\begin{proposition}
\label{prp:u}
Let the functions $(u_{n,k})$ be the the unique classical solution of the nested sequence of PDE problems \eqref{eq:u00.pde}--\eqref{eq:unk.pde}.  Then, omitting the spacial arguments $(x,y)$ to ease notation, the function $u_{0,0}$ is given by
\begin{align}
u_{0,0}(t)
	&=	\Pc_{0,0}(t,T) \varphi ,\label{eq:u00}
\end{align}
where $\Pc_{0,0}$ is defined in \eqref{eq:P00}, and for $(n,k)\neq(0,0)$, we have
\begin{align}
u_{n,k}(t)
	&=\sum_{j = 1}^{n+k} \sum_{I_{n,k,j}} \int_t^T \dd s_1 \int_{s_1}^T \dd s_2 \cdots \int_{s_{j-1}}^T \dd s_j\label{eq:unk} \\ &\quad
			\Pc_{0,0}(t,s_1) \Ac_{n_{1},k_{1}} 
			\Pc_{0,0}(s_1,s_2) \Ac_{n_{2},k_2}\cdots 
			\Pc_{0,0}(s_{j-1},s_{j}) \Ac_{n_j,k_j}
			\Pc_{0,0}(s_j,T) \varphi ,
\end{align}
with $I_{n,k,j}$ given by 
\begin{align}
I_{n,k,j} 
	= \left\{
	\begin{pmatrix} 
		n_1 ,\cdots, n_j  \\ 
	        k_1, \cdots, k_j           
	\end{pmatrix}
	\in \Zb_+^{2 \times j}
					\;\middle|\;
	\begin{aligned}
	& n_1 + \cdots +n_j = n,\\
	&k_1+\cdots+k_j = k, \\
	&1 \leq n_i + k_i, \text{ for all }
	1 \leq i \leq j
	\end{aligned} 
	\right\}. 
\end{align} 
\end{proposition}

\begin{proof}
See Appendix \ref{App:AppendixA}.
\end{proof}

For clarity, we present the lowest order terms here
\begin{align}
u_{1,0}(t)
	&=	\int_t^T \dd s_1 \, \Pc_{0,0}(t,s_1) \Ac_{1,0} \Pc_{0,0}(s_1,T) \varphi , \label{eq:u10.expl}\\
u_{0,1}(t)
	&=	\int_t^T \dd s_1 \, \Pc_{0,0}(t,s_1) \Ac_{0,1} \Pc_{0,0}(s_1,T) \varphi , \\
u_{2,0}(t)
	&=	\int_t^T \dd s_1 \int_{s_1}^T \dd s_2 \, \Pc_{0,0}(t,s_1) \Ac_{1,0} \Pc_{0,0}(s_1,s_2) \Ac_{1,0} \Pc_{0,0}(s_2,T) \varphi \\ &\quad
			+ \int_t^T \dd s_1 \, \Pc_{0,0}(t,s_1) \Ac_{2,0} \Pc_{0,0}(s_1,T) \varphi , \\
u_{1,1}(t)
	&= \int_t^T \dd s_1 \int_{s_1}^T \dd s_2 \, \Pc_{0,0}(t,s_1) \Ac_{0,1} \Pc_{0,0}(s_1,s_2) \Ac_{1,0} \Pc_{0,0}(s_2,T) \varphi \\ &\quad
			+ \int_t^T \dd s_1 \int_{s_1}^T \dd s_2 \, \Pc_{0,0}(t,s_1) \Ac_{1,0} \Pc_{0,0}(s_1,s_2) \Ac_{0,1} \Pc_{0,0}(s_2,T) \varphi \\ &\quad
			+ \int_t^T \dd s_1 \, \Pc_{0,0}(t,s_1) \Ac_{1,1} \Pc_{0,0}(s_1,T) \varphi , \\
u_{0,2}(t)
	&=	\int_t^T \dd s_1 \int_{s_1}^T \dd s_2 \, \Pc_{0,0}(t,s_1) \Ac_{0,1} \Pc_{0,0}(s_1,s_2) \Ac_{0,1} \Pc_{0,0}(s_2,T) \varphi \\&\quad
	                + \int_t^T \dd s_1\, \Pc_{0,0}(t,s_1) \Ac_{0,2} \Pc_{0,0}(s_1,T) \varphi.
\end{align}
To proceed further, we must specify explicitly the action of the semigroup $\Pc_{0,0}$.  We will consider three separate cases: European claims, single-barrier claims, and double-barrier claims.


\subsection{European claims}
In this section, we consider the case $I = \Rb$.  As $\tau = \infty$ when $I = \Rb$, we see from \eqref{eq:payoff} that this case corresponds to a European claim written on $X$.  We begin with the following lemma.

\begin{lemma}
\label{lem:eu}
Let $\Hc$ be the following linear operator
\begin{align}
\Hc
	&=	b \d_x + a \d_x^2 , &
\textup{dom}(\Hc)
	&=	C^2(\Rb) .
\end{align}
The following holds
\begin{align}
\Hc \psi_\om
	&=	\lam_\om \psi_\om ,
&\om
	 &\in \Rb,  \\
\psi_\om(x)
	&=	\frac{1}{\sqrt{2\pi}}\ee^{\ii \om x} ,
&\lam_\om(b,a)
	&=	b \ii \om - a \om^2 \label{eq:psi.lam}.
\end{align}
Moreover, we have
\begin{align}
\< \psi_\om, \psi_\gam \>
	&=	\del(\om - \gam) , &
\< f, g \>
	&:=	\int_\Rb \dd x \, \fb(x) g(x) , \label{eq:ortho.eu}
\end{align}
where $\del(\om-\gam)$ is a Dirac delta function and $\fb$ denotes the complex conjugate of $f$.
\end{lemma}

\begin{proof}
The lemma can easily be checked by direct computation.
\end{proof}

\begin{proposition}
\label{prp:P00.eu}
Let $\Pc_{0,0}$ be the semigroup generated by $\Ac_{0,0}$ with $\textup{dom}(\Ac_{0,0}) = C^2(\Rb^2)$.  Then we have
\begin{align}
\Pc_{0,0}(t,T) f
	&=	\int_\Rb \dd \om \int_\Rb \dd \gam \, \ee^{\Lam_{\om,\gam}(T-t)} \< \Psi_{\om,\gam}, f \> 
\Psi_{\om,\gam} , \label{eq:P00.eu} \\
\< f, g \>
	&:=	\int_\Rb \dd x \int_\Rb \dd y \, \fb(x,y) g(x,y) ,
\end{align}
where $\Psi_{\om,\gam}$ and $\Lam_{\om,\gam}$ are given by
\begin{align}
\Psi_{\om,\gam} (x,y)
	&=	\psi_\om(x) \psi_\gam(y) , &
\Lam_{\om,\gam}
	&=	\lam_\om(b_1,a_1) + \lam_\gam(b_2,a_2) , \\
(b_1, a_1)
	&=	\(\mu_0,(\thalf\sig^2)_0\), &
(b_2, a_2)
	&=	\(c_0,(\thalf g^2)_0\) ,
\end{align}
with $\psi_\om$ and $\lam_\om(b,a)$ as defined in \eqref{eq:psi.lam}.
\end{proposition}

\begin{proof}
Using Lemma \ref{lem:eu} and
\begin{align}
\int_\Rb \dd \om \, \overline{\psi}_\om(x) \psi_\om(y)
	&=	\del(x-y) ,
\end{align}
one can check by direct computation that 
\begin{align}
\Gam_{0,0}(t,x,y;T;\xi,\eta)
	&:=	\int_\Rb \dd \om \int_\Rb \dd \gam \, \ee^{\Lam_{\om,\gam}(T-t)} \Psi_{\om,\gam}(x,y) 
\Psi_{\om,\gam}(\xi,\eta) , \label{eq:Gam00.eu}
\end{align}
satisfies \eqref{eq:Gam00.pde} and is therefore the fundamental solution of $(\d_t + \Ac_{0,0})$.  Expression \eqref{eq:P00.eu} follows directly by inserting \eqref{eq:Gam00.eu} into \eqref{eq:P00}.
\end{proof}

Now, from Proposition \ref{prp:u}, we see that the $(u_{n,k})$ are a sum of terms of the form
\begin{align}
A
	&= \( \prod_{i=1}^j  \int_{s_{i-1}}^T \dd s_i \, \Pc_{0,0}(s_{i-1},s_{i}) \Ac_{n_i,k_i} \) \Pc_{0,0}(s_j,T) 
\varphi  && 
	s_0 = t. \label{eq:Aset}
\end{align}
Using \eqref{eq:P00.eu}, we can write these terms as
\begin{align}
A
	&=	\int_\Rb \dd \om_{j+1} \int_\Rb \dd \gam_{j+1} \( \prod_{i=1}^j \int_{s_{i-1}}^T \dd s_i \int_\Rb \dd \om_i \int_\Rb \dd \gam_i \,
		\ee^{\Lam_{\om_i,\gam_i}(s_i-s_{i-1})} \< \Psi_{\om_i,\gam_i} , \Ac_{n_i,k_i} \Psi_{\om_{i+1},\gam_{i+1}} \> \) \\ &\quad
		\< \Psi_{\om_{j+1},\gam_{j+1}} , \varphi \> \ee^{\Lam_{\om_{j+1},\gam_{j+1}}(T-s_j)} \Psi_{\om_1,\gam_1} . \label{eq:A.eu}
\end{align}
Although the multiple integral may seem unwieldy, we shall see that all but a single integral collapses when we compute the elements
\begin{align}
\< \Psi_{\om_i,\gam_i} , \Ac_{n_i,k_i} \Psi_{\om_{i+1},\gam_{i+1}} \> ,
\end{align}
which appear in \eqref{eq:A.eu}

\begin{lemma}
\label{lem:bracket.eu}
Let $\Psi_{\om,\gam}$ and $\< \cdot, \cdot \>$ be as defined in Proposition \ref{prp:P00.eu}.  Define the operator
\begin{align}
\Bc
	&:= x^i y^j \d_x^k \d_y^l , \label{eq:B}
\end{align}
where $i,j,k,l \in \Zb_+$.  Then we have
\begin{align}
\< \Psi_{\om',\gam'}, \Bc \Psi_{\om,\gam} \>
	&=	(\ii \om)^k (\ii \gam)^l (- \ii \d_\om)^i (-\ii \d_\gam)^j \del(\om-\om') \del(\gam-\gam') .
\end{align}
\end{lemma}

\begin{proof}
The proof is a straightforward computation.  Recalling that $\Psi_{\om,\gam} = \tfrac{1}{2\pi}\ee^{\ii \om x + \ii \gam y}$, we have
\begin{align}
\< \Psi_{\om',\gam'}, \Bc \Psi_{\om,\gam} \>
	&=	\int_\Rb \dd x \int_\Rb \dd y \, \overline{\Psi}_{\om',\gam'}(x,y) x^i y^j \d_x^k \d_y^l \Psi_{\om,\gam}(x,y) \\
	&=	(\ii \om)^k (\ii \gam)^l \int_\Rb \dd x \int_\Rb \dd y \, \overline{\Psi}_{\om',\gam'}(x,y) x^i y^j \Psi_{\om,\gam}(x,y) \\
	&=	(\ii \om)^k (\ii \gam)^l (- \ii \d_\om)^i (-\ii \d_\gam)^j \int_\Rb \dd x \int_\Rb \dd y \, \overline{\Psi}_{\om',\gam'}\Psi_{\om,\gam}(x,y) \\
	&=	(\ii \om)^k (\ii \gam)^l (- \ii \d_\om)^i (-\ii \d_\gam)^j \del(\om-\om') \del(\gam-\gam') .
\end{align}
where, in the last step, we have used \eqref{eq:ortho.eu}.
\end{proof}

\begin{remark}
\label{rmk:d.delta}
The derivative of a Dirac delta function is defined as follows:
\begin{align}
\int_\Rb \dd \om \, f(\om) \d_\om^n \del(\om-\om')
	&=	\int_\Rb \dd \om \, \del(\om-\om') (-\d_\om)^n  f(\om) 
	=		(-\d_{\om'})^n  f(\om') , \label{eq:d.delta}
\end{align}
where we have integrated by parts.
\end{remark}

Note from \eqref{eq:An0}, \eqref{eq:An1} and \eqref{eq:gn} that the operators $(\Ac_{n,k})$ are sums of operators of the form \eqref{eq:B}.  Thus, in light of Lemma \ref{lem:bracket.eu} and Remark \ref{rmk:d.delta}, we see that the integrals in \eqref{eq:A.eu} with respect to $\om_i$ and $\gam_i$, for $i = 1, 2, \ldots j$ collapse due to the Dirac delta functions.  The integral with respect to $\gam_{j+1}$ also collapses, due to the fact that the payoff function $\varphi$ does not depend on $y$.  And the iterated integrals with respect to $s_i$ for $i = 1, 2, \ldots j$ involve only exponentials and can always be computed explicitly.  Thus, what remains is the integral with respect to $\om_{j+1}$, which, in general, must be computed numerically (if $\varphi(x) = x^n \ee^{px}$ for some $n \in \Zb_+$ and $p \in \Rb$, then the integral with respect to $\om_{j+1}$ can be evaluated analytically).


\subsection{Single-barrier claims}
\label{sec:sb}
In this section, we consider the case $I = (L,\infty)$, which corresponds to a single-barrier knock-out claim written on $X$ with a barrier $L<X_0$.  
The case $I = (-\infty,U)$ with $U > X_0$ can be handled analogously.  We begin with the following lemma.

\begin{lemma}
\label{lem:sb}
Let $\Hc$ be the following linear operator
\begin{align}
\Hc
	&=	b \d_x + a \d_x^2 , &
\textup{dom}(\Hc)
	&=	\{ f \in C^2(I) : \lim_{x \to L} f(x) = 0 \} , &
I
	&=	(L,\infty).
\end{align}
The following holds
\begin{align}
\Hc \eta_\om
	&=	\mu_\om \eta_\om , &
\om
	&\in \Rb_+, \\
\eta_\om(x;b,a)
	&=	\sqrt{\frac{2}{\pi}}\ee^{-b x/(2a)} \sin \Big( \om(x-L) \Big) , &
\mu_\om(b,a)
	&=	- \frac{b^2}{4a} - a \om^2. \label{eq:eta.mu}
\end{align}
Moreover, we have
\begin{align}
\< \eta_\om, \eta_\gam \>
	&=	\del(\om - \gam) , &
\< f, g \>
	&:=	\int_I \dd x \, f(x) g(x) \mf(x) , &
\mf(x)
	&=	\ee^{bx/a} . \label{eq:m}
\end{align}
Here, $\del(\om-\gam)$ is a Dirac delta function.
\end{lemma}

\begin{proof}
The lemma can be checked by direct computation.
\end{proof}

\begin{proposition}\label{prp:P00.sb}
Let $\Pc_{0,0}$ the semigroup generated by $\Ac_{0,0}$ with 
\begin{align}
\textup{dom}(\Ac_{0,0}) 
	&= \{ f \in C^2(E) : \lim_{x \to L}f(x,y)=0 \} , &
E 
	&= (L,\infty) \times \Rb .
\end{align}
Then we have
\begin{align}
\Pc_{0,0}(t,T) f
	&=	\int_{\Rb_+}  \dd \om \int_\Rb \dd \gam \, \ee^{\Lam_{\om,\gam}(T-t)} \< \Psi_{\om,\gam}, f \> \Psi_{\om,\gam} , \label{eq:P00.sb} \\
\< f, g \>
	&:=	\int_I \dd x \int_\Rb \dd y \, \fb(x,y) g(x,y) m(x), &
I
	&=	(L,\infty) ,
\end{align}
where $\Psi_{\om,\gam}$ and $\Lam_{\om,\gam}$ are given by 
\begin{align}
\Psi_{\om,\gam} (x,y)
	&=	\eta_\om(x;b_1,a_1) \psi_\gam(y) , &
\Lam_{\om,\gam}
	&=	\mu_\om(b_1,a_1) + \lam_\gam(b_2,a_2) , \\
(b_1, a_1)
	&=	\(\mu_0,(\thalf \sig^2)_0\), &
(b_2, a_2)
	&=	\(c_0,(\thalf g^2)_0\) ,
\end{align}
with $\psi_\gam$ and $\lam_\gam$ as defined in \eqref{eq:psi.lam}, $\eta_\om$ and $\gam_\om$ as defined in \eqref{eq:eta.mu} and $m$ as defined in \eqref{eq:m}.
\end{proposition}

\begin{proof}
Using Lemma \ref{lem:sb}, we check by direct computation that
\begin{align}
\Gam_{0,0}(t,x,y;T;\xi,\eta)
	&:=	\int_{\Rb_+}\dd \om \int_\Rb \dd \gam \, \ee^{\Lam_{\om,\gam}(T-t)} \Psi_{\om,\gam}(x,y) 
		\Psi_{\om,\gam}(\xi,\eta) \mf(\xi) , \label{eq:Gam00.sb}
\end{align}
satisfies \eqref{eq:Gam00.pde} and is therefore the fundamental solution of $(\d_t + \Ac_{0,0})$.  Expression \eqref{eq:P00.sb} follows directly by inserting \eqref{eq:Gam00.sb} into \eqref{eq:P00}.
\end{proof}

Once again, to compute the $(u_{n,k})$ , we must examine terms of the form \eqref{eq:Aset}.  Using \eqref{eq:P00.sb}, we write these terms as
\begin{align}
A
	&=	\int_\Rb \dd \om_{j+1} \int_\Rb \dd \gam_{j+1} \( \prod_{i=1}^j \int_{s_{i-1}}^T \dd s_i \int_\Rb \dd \om_i \int_\Rb \dd \gam_i \,
		\ee^{\Lam_{\om_i,\gam_i}(s_i-s_{i-1})} \< \Psi_{\om_i,\gam_i} , \Ac_{n_i,k_i} \Psi_{\om_{i+1},\gam_{i+1}} \> \) \\ &\quad
		\< \Psi_{\om_{j+1},\gam_{j+1}} , \varphi \> \ee^{\Lam_{\om_{j+1},\gam_{j+1}}(T-s_j)} \Psi_{\om_1,\gam_1} , \label{eq:A.sb}
\end{align}
where $\Psi_{\om_i,\gam_i}$ and $\Lambda_{\om_i,\gam_i}$ are as in Proposition \ref{prp:P00.sb}. 
Noting that each $\Ac_{n,k}$ can be expressed as a sum of operators with the form of $\Bc$, which is defined  in \eqref{eq:B}, we must compute inner products of the form
\begin{align}
\< \Psi_{\om',\gam'}, \Bc \Psi_{\om,\gam} \>. 
\end{align}
This motivates the following lemma.

\begin{lemma}
\label{lem:ip.sb}
Let $\Psi_{\om,\gam}$ and $\< \cdot, \cdot \>$ be as defined in Proposition \ref{prp:P00.sb} and $\Bc$ be defined as 
in \eqref{eq:B}. Then 
\begin{align}
\< \Psi_{\om',\gam'}, \Bc \Psi_{\om,\gam} \>
	&= (\ii \gam)^l  (-\ii \d_\gam)^j \del(\gam-\gam') C_{\om',\om,i,k}, 
&C_{\om',\om,i,k}
	&=
	\sum_{m=0}^i \binom{i}{m} L^{i-m} C_{\om',\om,m,k}^{(1)}, \label{eq:ip.sb}
\end{align}
where
\begin{align}
C_{\om',\om,m,k}^{(1)}
	&= \frac{m!}{\pi}c_{\om,k}^{(e)} \( (\om + \om')^{-m-1} - |\om - \om'|^{-m-1}\)
		)\sin\( \frac{m\pi}{2}\) \\ &\quad
	+ \frac{m!}{\pi}c_{\om,k}^{(o)}  \( (\om + \om')^{-m-1} - \operatorname{sign}(\om - \om') |\om - \om'|^{-m-1}\) 
		\cos\(\frac{m\pi}{2}\),   \label{eq:C.sb}\\
c_{\om,k}^{(o)} 
	&=  
		 \sum_{m=1}^{\lfloor \frac{k+1} {2}\rfloor} \binom{k}{2m -1} (-1)^{k - m} 
		\(\frac{b}{2 a} \)^{k-2m+1} \om^{2m -1}, \label{eq:sb.codd}\\
c_{\om,k}^{(e)}
	&=  \sum_{m=0}^{\lfloor \frac{k}{2} \rfloor} \binom{k}{2m} (-1)^{k-m}
		\( \frac{b}{2 a} \)^{k - 2m} \om^{2m}.\label{eq:sb.ceven}
\end{align}
\end{lemma}

\begin{proof}
Recalling the definition of $\Psi$ from Proposition \ref{prp:P00.eu}, we compute
\begin{align}
\<  \Psi_{\om', \gam'}, \Bc \Psi_{\om,\gam}\> 
	&= \int_I \dd x \int_\Rb \dd y\ \mf(x) \overline{\Psi}_{\om',\gam'}(x,y)	
		x^i y^j \d_x^k \d_y^l \Psi_{\om,\gam}(x,y)  \\
	&= (\ii \gam)^l  (-\ii \d_\gam)^j \del(\gam-\gam') 
		\int_{I} \dd x\ \mf(x) x^i \eta_{\om'}(x) \d_x^k \eta_\om(x). \label{eq:xint.sb}
\end{align}
By direct computation, we find that 
\begin{align}
\d_x^k \eta_\om (x)
	= \sqrt{\frac{2}{\pi}} \ee^{-\frac{b }{2 a} x}
		\(
		c_{\om,k}^{(o)} \cos \( \om ( x - L )\) 
		+ c_{\om,k}^{(e)} \sin \(\om ( x -L ) \) 
		\), \label{eq:etak.sb}
\end{align}
where $c_{\om,k}^{(o)}$ and $c_{\om,k}^{(e)}$ are given by \eqref{eq:sb.codd} and \eqref{eq:sb.ceven}, respectively. Thus,
\begin{align}
\int_{I} \dd x\ \mf(x) x^i \eta_{\om'}(x) \d_x^k \eta_\om(x)
	&= \frac{2}{\pi} \int_I \dd x\ x^i 
		\sin\( \om' ( x - L ) \) 
		\( c_{\om,k}^{(o)} \cos\( \om (x - L ) \)  
		 + c_{\om,k}^{(e)} \sin\( \om (x - L ) \)
		 \) \\
	&=
\sum_{m=0}^i \binom{i}{m} L^{i-m}
\frac{2}{\pi} \int_0^\infty   \dd x\ x^m \sin(\om' x) \( c_{\om,k}^{(o)} \cos(\om x) + c_{\om,k}^{(e)} \sin(\om x) \) \\
	&= \sum_{m=0}^i \binom{i}{m} L^{i-m} C_{\om',\om,m,k}^{(1)} , \label{eq:matt}
\end{align}
where $C_{\om',\om,m,k}^{(1)}$ is given by \eqref{eq:C.sb}.
Inserting \eqref{eq:matt} into \eqref{eq:xint.sb} gives \eqref{eq:ip.sb}.
\end{proof}

Note from \eqref{eq:An0}, \eqref{eq:An1} and \eqref{eq:gn}, the operators $\Ac_{n,k}$ are sums of operators of the form \eqref{eq:B}. We see from \eqref{eq:ip.sb} that the integrals with respect to $\gamma_i$, $i = 1,2,3,\cdots,j$ in \eqref{eq:A.sb} collapse due to the Dirac delta functions. As $\varphi$ is independent of $y$, the integral with respect to $\gam_{j+1}$ also collapses. Furthermore, the iterated integrals with respect to $s_i$, $i= 1,2,3,\cdots j$ in \eqref{eq:A.sb} involve only exponentials in $s_i$ and can therefore be evaluated explicitly. We are left only with integrals with respect to $\om_{i}$, $i =1,2,3,\cdots,j+1$, which can be evaluated numerically.


\subsection{Double-barrier claims}
\label{sec:db}
In this section, we consider the case $I = (L,U)$, which corresponds to a double-barrier knock-out claim written on $X$ with a barriers $L$ and $U$ satisfying $L<X_0<U$.  We begin with the following lemma.

\begin{lemma}
\label{lem:db}
Let $\Hc$ be the following linear operator
\begin{align}
\Hc
	&=	b \d_x + a \d_x^2 , &
\textup{dom}(\Hc)
	&=	\{ f \in C^2(I) : \lim_{x \to L} f(x) = 0, \lim_{x \to U} f(x) = 0 \} , &
I
	&=	(L,U).
\end{align}
The following holds
\begin{align}
\Hc \phi_\ell
	&=	\nu_\ell \phi_\ell , &
\ell
	&\in \Nb \\
\phi_\ell(x;b,a)
	&= \sqrt{\frac{2}{U-L}} \ee^{ -\frac{b x}{2 a}} \sin \left(\frac{\pi  \ell (x-L)}{U-L}\right), &
\nu_\ell(b,a)
	&=	-\frac{b^2}{4 a}-\frac{a \pi ^2  \ell^2}{(U-L)^2}. \label{eq:phi.nu}
\end{align}
Moreover, we have
\begin{align}
\< \phi_\ell,\phi_k \>
	&=	\del_{\ell,k} , &
\< f, g \>
	&:=	\int_I \dd x \, \mf(x) f(x) g(x) . 
\end{align}
Here, $\del_{\ell,k}$ is a Kronecker delta function and $\mf$ is given by \eqref{eq:m}.
\end{lemma}

\begin{proof}
The lemma can be checked by direct computation.
\end{proof}

\begin{proposition}\label{prp:P00.db}
Let $\Pc_{0,0}$ the semigroup generated by $\Ac_{0,0}$ with 
\begin{align}
\textup{dom}(\Ac_{0,0}) 
	&= \{ f \in C^2(E) : \lim_{x \to L}f(x,y)=0 , \lim_{x \to U} f(x,y) = 0 \} , &
E 
	&= (L,R) \times \Rb .
\end{align}
Then we have
\begin{align}
\Pc_{0,0}(t,T) f
	&=	\sum_{\ell=1}^\infty  \int_\Rb \dd \gam \, \ee^{\Lam_{\ell,\gam}(T-t)} \< \Psi_{\ell,\gam}, f \> 
\Psi_{\ell,\gam} , 
\label{eq:P00.db} \\
\< f, g \>
	&:=	\int_I \dd x  \int_\Rb \dd y \,  \fb(x,y) g(x,y) \mf(x) , &
I
	&=	(L,U) ,
\end{align}
where $\Psi_{\ell,\gam}$ and $\Lam_{\ell,\gam}$ are given by
\begin{align}
\Psi_{\ell,\gam} (x,y)
	&=	\phi_\ell(x;b_1,a_1) \psi_\gam(y) , &
\Lam_{\ell,\gam}
	&=	\nu_\ell(b_1,a_1) + \lam_\gam(b_2,a_2), \\
(b_1, a_1)
	&=	\(\mu_0,(\thalf\sig^2)_0\), &
(b_2, a_2)
	&=	\(c_0,(\thalf g^2)_0\) ,
\end{align}
and $\psi_\gam$ and $\lam_\gam(b,a)$ are given in \eqref{eq:psi.lam}, $\phi_\ell$ and $\nu_\ell$ are given in 
\eqref{eq:phi.nu} and $\mf$ is given in \eqref{eq:m}.
\end{proposition}

\begin{proof}
Using Lemma \ref{lem:db}, we check by direct computation that
\begin{align}
\Gam_{0,0}(t,x,y;T;\xi,\eta)
	&:=	\sum_{\ell=1}^\infty \int_\Rb \dd \gam \, \ee^{\Lam_{\ell,\gam}(T-t)} \Psi_{\ell,\gam}(x,y) 
\Psi_{\ell,\gam}(\xi,\eta) \mf(\xi) , \label{eq:Gam00.db}
\end{align}
satisfies \eqref{eq:Gam00.pde} and is therefore the fundamental solution of $(\d_t + \Ac_{0,0})$.  Expression \eqref{eq:P00.db} follows directly by inserting \eqref{eq:Gam00.db} into \eqref{eq:P00}.
\end{proof}

As with the European and single-barrier cases, to compute the functions $(u_{n,k})$ we must evaluate terms of the form \eqref{eq:Aset}. Using \eqref{eq:P00.db} we write these terms as
\begin{align}
A
	&=	\sum_{\ell_{j+1}=1}^\infty \int_\Rb \dd \gam_{j+1} \( \prod_{i=1}^j \int_{s_{i-1}}^T \dd s_i 
		\sum_{\ell_{i} =1}^\infty \int_\Rb \dd \gam_i \,
		\ee^{\Lam_{\ell_i,\gam_i}(s_i-s_{i-1})} \< \Psi_{\ell_i,\gam_i} , \Ac_{n_i,k_i} 
\Psi_{\ell_{i+1},\gam_{i+1}} \> \) \\ &\quad
		\< \Psi_{\ell_{j+1},\gam_{j+1}} , \varphi \> \ee^{\Lam_{\ell_{j+1},\gam_{j+1}}(T-s_j)} 
\Psi_{\om_1,\gam_1} 
. \label{eq:A.db}
\end{align}
As each $\Ac_{n,k}$ is a sum of operators with the form of $\Bc$, which is defined in \eqref{eq:B}, we must compute terms of the form $\< \Psi_{\ell_i,\gam_i} , \Bc \Psi_{\ell_{i+1},\gam_{i+1}} \>$.

\begin{lemma}
\label{lem:ip.db}
Let $\Psi_{\ell,\gam}$ and $\< \cdot, \cdot \>$ be as defined in Proposition \ref{prp:P00.db} and $\Bc$ be as defined
in \eqref{eq:B}. Then
\begin{align}
\< \Psi_{\ell',\gam'}, \Bc \Psi_{\ell,\gam}\> 
	= (\ii \gam)^l  (-\ii \d_\gam)^j \del(\gam-\gam') C_{\ell',\ell,i,k} \label{eq:ip.db},
\end{align}
where
\begin{align}
C_{\ell',\ell,i,k}
	=
\sum_{m=0}^i \binom{i}{m} L^{i-m}\( \frac{U-L}{\pi}\)^{m+1}  \( C^{(1)}_{\ell',\ell,m,k} \Ib_{\ell' \neq \ell} + 
C^{(2)}_{\ell,m,k} \del_{\ell',\ell}\), \label{eq:c.db}
\end{align}
and
\begin{align}
C^{(1)}_{\ell',\ell,m,k} 
	&= 
	\frac{1}{4} \pi ^{m+\frac{3}{2}} \Gam_E \left(\frac{m+1}{2}\right) c_{\ell,k}^{(e)} 
		\, _1\widetilde{F}_2\left(\frac{m+1}{2};\frac{1}{2},\frac{m+3}{2};-\frac{1}{4} (\ell-\ell')^2 \pi ^2\right) \label{eq:db.c1}\\&\quad
	-\frac{1}{4} \pi ^{m+\frac{3}{2}} \Gam_E \left(\frac{m+1}{2}\right) c_{\ell,k}^{(e)} 
		\, _1\widetilde{F}_2\left(\frac{m+1}{2};\frac{1}{2},\frac{m+3}{2};-\frac{1}{4} (\ell+\ell')^2 \pi ^2\right)\\&\quad
	+\frac{\pi ^{m+2}}{2 (m+2)}\left(\ell'  c_{\ell,k}^{(o)}-\ell  c_{\ell,k}^{(o)}\right)
		\, _1F_2\left(\frac{m}{2}+1;\frac{3}{2},\frac{m}{2}+2;-\frac{1}{4} (\ell-\ell')^2 \pi ^2\right)  \\&\quad
	+\frac{\pi ^{m+2}}{2 (m+2)}\left(\ell'  c_{\ell,k}^{(o)}+\ell  c_{\ell,k}^{(o)}\right)
		\, _1F_2\left(\frac{m}{2}+1;\frac{3}{2},\frac{m}{2}+2;-\frac{1}{4} (\ell+\ell')^2 \pi ^2\right), \\
C^{(2)}_{\ell,m,k}
	&= 
\frac{\pi ^{m+2} \ell c_{l,k}^{\text{(o)}} }{m+2}\, _1F_2\left(\frac{m}{2}+1;\frac{3}{2},\frac{m}{2}+2;-\pi ^2 \ell^2\right) \label{eq:db.c2}\\ &\quad
-\frac{2 \pi ^{m+3} \ell^2 c_{l,k}^{\text{(e)}} }{m^2+4 m+3} \, _1F_2\left(\frac{m}{2}+\frac{3}{2};\frac{3}{2},\frac{m}{2}+\frac{5}{2};-\pi ^2 \ell^2\right) \\
c_{\ell,k}^{(o)}
	&=  
		\sqrt{\frac{2}{U-L}} \sum_{j=1}^{\lfloor \frac{k+1} {2}\rfloor} \binom{k}{2j -1} (-1)^{k - j} 
		\(\frac{b}{2 a} \)^{k-2j+1} \(\frac{\pi \ell}{U-L}\)^{2j -1}, \\
c_{\ell,k}^{(e)}
	&=\sqrt{\frac{2}{U-L}}  \sum_{j=0}^{\lfloor \frac{k}{2} \rfloor} \binom{k}{2j} (-1)^{k-j}
		\( \frac{b}{2 a} \)^{k - 2j} \(\frac{\pi \ell}{U-L}\)^{2j}.
\end{align}
Here, $\Gam_E$ is the Euler gamma function, and $_1 \mathrm{F}_2$ and $_1\widetilde{\mathrm{F}}_2$ are hypergeometric 
and regularized hypergeometric functions, respectively.
\end{lemma}
\begin{proof}
Recalling the definition of $\Psi$ from Proposition \ref{prp:P00.db}, we compute
\begin{align}
\<  \Psi_{\ell', \gam'}, \Bc \Psi_{\ell,\gam}\> 
	&= \int_I \dd x \int_\Rb \dd y\ \mf(x) \overline{\Psi}_{\ell',\gam'}(x,y)	
		x^i y^j \d_x^k \d_y^j \Psi_{\ell,\gam}(x,y)  \\
	&= (\ii \gam)^l  (-\ii \d_\gam)^j \del(\gam-\gam') 
		\int_{I} \dd x\ \mf(x) x^i \phi_{\ell'}(x) \d_x^k \phi_\ell(x) \label{eq:ipint.db}.
\end{align}
We see by direct computation that 
\begin{align}
\d_x^k  \phi_\ell(x; b,a)
	=  \ee^{-\frac{b }{2 a} x}
		\(
		c_{\ell,k}^{(o)} \cos \( \frac{ \ell \pi(x - L)}{U-L} \) 
		+ c_{\ell,k}^{(e)} \sin \( \frac{\ell \pi (x - L)}{U-L} \) 
		\) . \label{eq:phik.db}
\end{align}
Thus, we have
\begin{align}
&\int_{I} \dd x\ \mf(x) x^i \phi_{\ell'}(x) \d_x^k \phi_\ell(x) \\
	&= \int_I \dd x\ x^i \sin \( \frac{\ell' \pi (x-L)}{U-L} \) 
	\( c_{\ell,k}^{(o)} \cos\( \frac{\ell \pi (x-L)}{U-L} \) 
	  +c_{\ell,k}^{(e)} \sin\( \frac{\ell \pi (x-L)}{U-L} \) \) \\
	&=  
	\sum_{m=0}^i \binom{i}{m} L^{i-m}\( \frac{U-L}{\pi}\)^{m+1} 
		\int_0^\pi \dd x\ x^m \sin (\ell' x) 
		\( c_{\ell,k}^{(o)} \cos( \ell x) + c_{\ell,k}^{(e)} \sin ( \ell x) \) \\ 
	&= 
	\sum_{m=0}^i \binom{i}{m} L^{i-m}\( \frac{U-L}{\pi}\)^{m+1} 
		\( C^{(1)}_{\ell',\ell,m,k} \Ib_{\ell' \neq \ell} + 
			C^{(2)}_{\ell,m,k} \del_{\ell',\ell}\), \label{eq:db.intfinal}
\end{align}
where the formulas for $C^{(1)}_{\ell',\ell,m,k}$ and $C^{(2)}_{\ell',\ell,m,k}$ are given in \eqref{eq:db.c1} and \eqref{eq:db.c2}, respectively. Inserting \eqref{eq:db.intfinal} into \eqref{eq:ipint.db} yields \eqref{eq:ip.db}. 
\end{proof}

\begin{remark}
The functions $_1 F_2$ and $_1 \widetilde F_2$, which appear in the expression for $C_{\ell',\ell,i,k}$, arise from computing integrals of the form
$\int_0^\pi\dd x\, x^m \sin(\ell ' x) \cos(\ell x)$ and $\int_0^\pi\dd x\, x^m \sin(\ell ' x) \sin(\ell x)$.
For any $m,\ell,\ell' \in \Nb_0$, these integrals are equal to finite sums of terms containing powers of $x$, sines and cosines (as can be seen by integrating by parts).  Thus, the functions $_1 F_2$ and $_1 \widetilde F_2$ can be evaluated with minimal computational effort.
\end{remark}

Note from \eqref{eq:An0}, \eqref{eq:An1} and \eqref{eq:gn}, the operators $\Ac_{n,k}$ are sums of operators of the form \eqref{eq:B}. We see from \eqref{eq:ip.db} that the integrals with respect to $\gamma_i$, $i = 1,2,3,\cdots,j$ in \eqref{eq:A.db} collapse due to the Dirac delta functions. Since $\varphi$ is independent of $y$, the integral with respect to $\gam_{j+1}$ also collapses. Furthermore, the iterated integrals with respect to $s_i$, $i= 1,2,3,\cdots j$ in \eqref{eq:A.sb} involve only exponentials in $s_i$ and can therefore be evaluated explicitly. Thus, \eqref{eq:A.db} is an explicit sum and does not require any numerical integration.

\begin{remark}
The fundamental solution corresponding to the parabolic operator $(\d_t + \Ac_{0,0} + \rho \Ac_{0,1})$ can be obtained explicitly in all three of the cases we have considered (European, single barrier and double barrier).  As such, one might wonder why we expand the operator $\Ac$ in powers of $(x-\xb)$ and $(y-\yb)$ as well as in powers $\rho$ (as opposed to expanding in powers of $(x-\xb)$ and $(y-\yb)$ only).  The reason we expand in powers of $\rho$ is that, without this expansion, the integrals in \eqref{eq:unk} with respect to $s_1, s_2, \ldots, s_j$ cannot be computed explicitly in the single or double-barrier cases.  Thus, by expanding in $\rho$ avoid having to evaluate multidimensional numerical integrals.
\end{remark}

%
%

\section{Accuracy results}
\label{sec:accuracy}
In this section, we establish the accuracy of our formal pricing approximation for European options. Before stating our accuracy result, let us introduce some additional notation. For a set $E \subset \Rb^d$, denote by $C_b^{n,1}(E)$ the class of bounded functions on $E$ with globally Lipschitz continuous derivatives of order less than or equal to $n$. Let $\| f\|_{C_b^{n,1}}$ denote the sum of the $L^\infty$-norms of the derivatives of $f$ up to order $n$. We also denote by $C_b^{-1,1}(E) = L^\infty (E)$ and we set $\| \cdot \|_{C_b^{-1,1}} = \| \cdot \|_{L^\infty}$.
The following theorem describes the accuracy of the $N$th order approximation of the price of a European option written on an asset described by local-stochastic volatility dynamics.

\begin{theorem}
\label{thm:e.acc}
Consider the case $I = \Rb$. Suppose for some non-negative integer $N$ that $\sig, \mu,c,g \in C_b^{N,1}(\Rb^2)$ and that there exists a positive constant $M$ such that 
\begin{align}
\frac{1}{M} \leq \| \sig\|_{C_b^{N,1}},\, \| \mu \|_{C_b^{N,1}},\, \| c \|_{C_b^{N,1}},\, \| g \|_{C_b^{N,1}} \leq M .
\end{align}
Furthermore, assume that $\varphi \in C_b^{h-1,1}(\Rb^2)$ for some $0 \leq h \leq 2$. Then we have 
\begin{align}
|(u-\ub_0^\rho)(t,x,y)|
	&\leq C\, (T-t)^{\frac{h+1}{2}}, \label{eq:e.zeroaccuracy}
&& 0 \leq t < T,
&& x\in I, y \in \Rb.
\end{align}
For $N \geq 1$, we have
\begin{align}
|(u-\ub_N^\rho)(t,x,y)|
	&\leq C\, ((T-t)^{\frac 12} + |\rho|) \sum_{i=0}^N  |\rho|^i (T-t)^{\frac{N-i+h}{2}}
&& 0 \leq t < T,
&& x \in I, y \in \Rb.
\label{eq:e.accuracy}
\end{align}
The positive constants $C$ in \eqref{eq:e.zeroaccuracy} and \eqref{eq:e.accuracy} depend only on $M,N$ and $\|\varphi\|_{C_b^{h-1,1}}$. 
\end{theorem}

\begin{proof}
See Appendix \ref{App:e.proof}.
\end{proof}

Establishing asymptotic accuracy for barrier-style claims remains an open problem for the following reason. The proof of Theorem \ref{thm:e.acc} exploits Gaussian symmetry present in the pricing kernel of the zeroth order European problem. This symmetry is absent in both the single-barrier and double-barrier cases, and hence the same techniques for proving accuracy cannot be applied.  In the following section we explore the accuracy of our approximations for barrier-style claims in several numerical examples.

%
%

\section{Numerical examples}
\label{sec:numerics}

\subsection{Heston model}
\label{sec:heston}
In this section, we implement our pricing approximation for an underlying $S = \ee^X$ that has Heston dynamics \cite{heston1993}.
Specifically, we suppose that $(X,Y)$ satisfies
\begin{align}
\dd X_t
	&= - \half Y_t\, \dd t + \sqrt{Y_t}\, \dd W_t,  &
\dd Y_t
	&= \kappa (\theta - Y_t)\, \dd t  + \delta \sqrt{Y_t}\, \dd B_t, &
\dd \<W,B\>_t 
	&=	\rho\, \dd t , \label{eq:plotdynamics}
\end{align}
where $2 \kappa \theta \geq \del^2$ so that the $Y$ process remains strictly positive.  In our numerical experiments, we consider double-barrier knock-out calls and puts with the following parameters fixed
\begin{center}
\begin{tabular}{|l|l|l|l|l|l|l|l|l|}
\hline 
	$X_0$ & $Y_0$ & $K$  & $T$ & $\rho$ & $\kappa$ & $\theta$ & $\delta$ \\
\hline
	0.62  & 0.04   & .62  & 0.083 & -0.4 & 1.15 & 0.04 & 0.2 \\
\hline
\end{tabular}
\end{center}
where $\ee^K$ represents the strike and $T$ represents the maturity date. We first consider call payoffs $\varphi(x) = (\ee^x - \ee^K)^+$ with the lower barrier $L = 0$ fixed and the upper barrier $U>K$ varying. We compute both our zeroth and second order price approximation $\ub_0^\rho$ and $\ub_2^\rho$ as well as ``exact'' price $u$, which we obtain via Monte Carlo simulation.    In Figure \ref{fig:ATMcallerror}, we plot the error $u-\ub_0^\rho$ and $u-\ub_2^\rho$ of our zeroth and second order approximations as a function of the upper barrier $U$.  To get a sense of the scale of the error, we also plot in Figure \ref{fig:ATMcall} the exact price $u$ as a function of $U$. In Figures \ref{fig:ATMputerror} and \ref{fig:ATMput}, we provide analogous plots for put payoffs $\varphi(x) = (\ee^K - \ee^x)^+$ with the upper barrier $U=1$ fixed while varying the lower barrier $L<K$. We see from Figures \ref{fig:ATMcallerror} and \ref{fig:ATMputerror} that $\bar u_2$ provides a more accurate approximation of $u$ than $\ub_0^\rho$ for both puts and calls at nearly all levels of $L$ and $U$. 

\begin{remark}
We omit the first order approximation $\ub_1^\rho$ in Figures \ref{fig:ATMcallerror} and \ref{fig:ATMputerror} for the following reason.  The difference $|\ub_0^\rho - \ub_1^\rho|$ is small compared to $|\ub_0^\rho - \ub_2^\rho|$ because, when the payoff function $\varphi$ depends only on $x$ (as is the case for call and put payoffs), we have $u_{0,1} = 0$. Therefore, the first correlation correction term in our approximation appears at the second order in $u_{1,1}$.  The effect of including the first correlation correction is large compared to the first correction due to $y$-dependence in the coefficients of $\Ac$.
\end{remark}

\subsection{CEV Model}
\label{sec:cev}
In this section, we implement our pricing approximation for an underlying $S = \ee^X$ that has Constant Elasticity of Variance (CEV) dynamics \cite{CoxCEV}.
Specifically, we suppose that $X$ satisfies
\begin{align}
\dd X_t
	&= - \frac 12 \sig^2 \ee^{2 X_t (\gam - 1)}\, \dd t + \sig \ee^{X_t (\gam - 1)}\, \dd W_t. \label{eq:CEVdynamics}
\end{align}
where $\sig >0$ and $\gam > 0 $.  We consider double-barrier knock-out calls and puts with the following parameters fixed
\begin{center}
\begin{tabular}{|l|l|l|l|l|}
\hline 
	$X_0$ & $K$  & $T$  & $\sig$ & $\gam$ \\
\hline
	0.62  & 0.62  & 0.083 & 0.32 & 0.019  \\
\hline
\end{tabular}
\end{center}
where $\ee^K$ represents the strike and $T$ represents the maturity date. We first consider call payoffs $\varphi(x) = (\ee^x - \ee^K)^+$ with the lower barrier $L = 0$ fixed and the upper barrier $U>K$ varying. We compute the zeroth and second order price approximation $\ub_0$ and $\ub_2$, respectively, as well as ``exact'' price $u$, which we obtain via Monte Carlo simulation.  Note that we omit the superscript $\rho$ from $u$ and $\ub$ as correlation plays no role in a local volatility setting.
In Figure \ref{fig:ATMCEVcallerror}, we plot the error $u-\ub_0$ and $u-\ub_2$ of our zeroth and second order approximations as functions of the upper barrier $U$.  To get a sense of the scale of the error, we also plot in Figure \ref{fig:ATMCEVcall} the exact price $u$ as a function of $U$. In Figures \ref{fig:ATMCEVputerror} and \ref{fig:ATMCEVput}, we provide analogous plots for put payoffs $\varphi(x) = (\ee^K - \ee^x)^+$ with the upper barrier $U=1$ fixed while varying the lower barrier $L<K$. 
We see from Figures \ref{fig:ATMCEVcallerror} and \ref{fig:ATMCEVputerror} that in both the call and put cases the second order approximation outperforms the zeroth order approximation.

\section{Conclusion}
\label{sec:conclusion}
In this paper we have presented a formal pricing approximation for European and barrier-style claims in a local-stochastic volatility setting.  We have provided rigorous accuracy results for European-style claims.  And we have provided several numerical examples illustrating the accuracy and versatility of our approximation for barrier-style claims.  Future research will focus on extending our techniques to other path-dependent derivatives, such as lookback and variance-style claims.

\subsection*{Acknowledgments}
The authors are grateful to Stefano Pagliarani and Andrea Pascucci for helpful feedback on this manuscript.

%
%

\appendix
\section{Proof of Proposition \ref{prp:u}} \label{App:AppendixA}
In this section, we present the proof of Proposition \ref{prp:u}. 

\begin{proof}[Proof of Proposition \ref{prp:u}]
We first note that formula \eqref{eq:unk} holds for $(n,k) = (1,0)$ and $(n,k) = (0,1)$ by applying Duhamel's principle to \eqref{eq:u10.pde} and \eqref{eq:u01.pde}. Next, assume as an inductive hypothesis that for non-negative integers $n$ and $k$ such that $n + k \geq 1$ formula \eqref{eq:unk} holds for pairs of non-negative integers $(m,j)$ such that $m+j \leq n+k$. Define 
\begin{align}
A_{n,k}^b
	:= \{ (i,j) \mid 0 \leq i \leq n,\, 0 \leq j \leq k,\, 1 \leq i + j \leq b \}.
\end{align}
Applying Duhamel's principle to \eqref{eq:unk.pde}, we see that
\begin{align}
&u_{n+1,k} \\
	= &\int_{t}^T \dd s \Pc_{0,0} (t,s) \( 
						\sum_{i=0}^{n+1} \sum_{j=0}^k (1-\del_{i+j,0}) \Ac_{i,j}u_{n-i+1,k-j} 
					\) 
	= \int_{t}^T \dd s \Pc_{0,0} (t,s) \( 
						\sum_{\substack{(i,j) \in\\ A_{n+1,k}^{n+k+1}}} \Ac_{i,j}u_{n-i+1,k-j} 
					\) \\
	=&
		\int_t^T \dd s \Pc_{0,0}(t,s) \Ac_{n+1,k} \Pc_{0,0}(s,T) \varphi 
		+\sum_{\substack{(i,j) \in\\ A_{n+1,k}^{n+k+1}}}\sum_{l=1}^{n+k-i-j+1} \sum_{I_{n-i+1,k-j,l}} 
			\int_t^T \dd s \int_s^T \dd s_1 \cdots \int_{s_{l-1}}^T \dd s_l \\
	&\quad 
			\Pc_{0,0}(t,s) \Ac_{i,j} 
			\Pc_{0,0}(s,s_1) \Ac_{(n-i+1)_1,(k-j)_1} \cdots 
			\Pc_{0,0}(s_{l-1},s_l) \Ac_{(n-i+1)_l,(k-j)_l} 
			\Pc_{0,0}(s_l,T)\varphi, \label{eq:ps1} 
\end{align}
where \eqref{eq:ps1} follows from our inductive hypothesis.  Reordering the sums in \eqref{eq:ps1} we obtain
\begin{align}
u_{n+1,k} 
	&= \int_t^T \dd s \Pc_{0,0}(t,s) \Ac_{n+1,k} \Pc_{0,0}(s,T) \varphi 
	+\sum_{l=1}^{n+k}\sum_{\substack{(i,j) \in\\ A_{n+1,k}^{n+k-\ell+1}}} \sum_{I_{n-i+1,k-j,l}} 
			\int_t^T \dd s \int_s^T \dd s_1 \cdots \int_{s_{l-1}}^T \dd s_l \\
	&\quad 
			\Pc_{0,0}(t,s) \Ac_{i,j} 
			\Pc_{0,0}(s,s_1) \Ac_{(n-i+1)_1,(k-j)_1} \cdots 
			\Pc_{0,0}(s_{l-1},s_l) \Ac_{(n-i+1)_l,(k-j)_l} 
			\Pc_{0,0}(s_l,T)\varphi. \label{eq:ps2}
\end{align}
Next, note that
\begin{align}
I_{n+1,k,\ell}
	= \bigcup_{(i,j) \in A_{n+1,k}^{n+k-\ell+1}} \left\{ 
	\begin{pmatrix}
	 i & n_1 & n_2 & \cdots & n_{\ell - 1} \\
	 j & k_1 & k_2 & \cdots & k_{\ell -1 }
	\end{pmatrix}
	\Big| 
	\begin{pmatrix}
	  n_1 & n_2 & \cdots & n_{\ell - 1} \\
	  k_1 & k_2 & \cdots & k_{\ell -1 }
	\end{pmatrix} \in I_{n-i+1,k-j,\ell -1}\right\}. \label{eq:inkform}
\end{align}
Therefore, combining \eqref{eq:inkform} with \eqref{eq:ps2} we obtain
\begin{align}
u_{n+1,k} 
	=&\int_t^T \dd s \Pc_{0,0}(t,s) \Ac_{n+1,k} \Pc_{0,0}(s,T) \varphi 
	+\sum_{l=1}^{n+k} \sum_{I_{n+1,k,l+1}} 
			\int_t^T \dd s \int_s^T \dd s_1 \cdots \int_{s_{l-1}}^T \dd s_l \\
	&\quad 
		\Pc_{0,0}(t,s) \Ac_{(n+1)_1,k_1} 
		\Pc_{0,0}(s,s_1) \Ac_{(n+1)_2,k_2} \cdots 
		\Pc_{0,0}(s_{l-1},s_l) \Ac_{(n+1)_{l+1},k_{l+1}} 
		\Pc_{0,0}(s_l,T)\varphi. 
\end{align}
Relabeling $(s,s_1,s_2,\cdots,s_\ell) \to (s_1,s_2,\cdots,s_{\ell+1})$ and reindexing and gives
\begin{align}
u_{n+1,k} 
	=&\int_t^T \dd s_1 \Pc_{0,0}(t,s_1) \Ac_{n+1,k} \Pc_{0,0}(s_1,T) \varphi 
	+\sum_{l=2}^{n+k+1} \sum_{I_{n+1,k,l}} 
			\int_t^T \dd s_1 \int_{s_1}^T \dd s_2 \cdots \int_{s_{l-1}}^T \dd s_{l} \\
	&\quad 
		\Pc_{0,0}(t,s_1) \Ac_{(n+1)_1,k_1} 
		\Pc_{0,0}(s_1,s_2) \Ac_{(n+1)_2,k_2} \cdots 
		\Pc_{0,0}(s_{l-1},s_{l}) \Ac_{(n+1)_{l},k_{l}} 
		\Pc_{0,0}(s_{l},T)\varphi \\
	=&\sum_{l=1}^{n+k+1} \sum_{I_{n+1,k,l}} 
			\int_t^T \dd s_1 \int_{s_1}^T \dd s_2 \cdots \int_{s_{\ell}}^T \dd s_{l} \\
	&\quad 
		\Pc_{0,0}(t,s_1) \Ac_{(n+1)_1,k_1} 
		\Pc_{0,0}(s_1,s_2) \Ac_{(n+1)_2,k_2} \cdots 
		\Pc_{0,0}(s_{l-1},s_{l}) \Ac_{(n+1)_{l},k_{l}} 
		\Pc_{0,0}(s_{l},T)\varphi, 
\end{align}
which is \eqref{eq:unk} for the case $(n+1,k)$. The proof for the case $(n,k+1)$ is analogous.
\end{proof}

\section{Proof of Theorem \ref{thm:e.acc}}
\label{App:e.proof}

In this section, we prove Theorem \ref{thm:e.acc}.  Our strategy is to adapt the proof of asymptotic accuracy in \cite{lorig-pagliarani-pascucci-4} to our present situation.  As such, many of the propositions and lemmas needed for the proof of Theorem \ref{thm:e.acc} follow from analogous propositions and lemmas contained in \cite{lorig-pagliarani-pascucci-4}.

Throughout this section, we let $z = (x,y)$, $\zb = (\xb,\yb)$ and $\zeta = (\xi,\eta)$ be elements of $\Rb^2$. It will also be convenient to introduce multi-index notation for the operators $\Ac$ and $\Ac_{n,k}$.  We have 
\begin{align}
\Ac	
	&= \sum_{|\alpha| \leq 2} a_\alpha(z) D_z^\alpha, &
\Ac_{n,k}
	&= \sum_{\alpha \in A_k} a_{\alpha,n} (z) D_z^\alpha, \label{def:multiA}\\
A_0
	&= \{ (1,0), (0,1), (2,0), (0,2)\}, &
A_1
	&= \{ (1,1) \},
\end{align}
where
\begin{align}
\alpha
	&= (\alpha_1,\alpha_2) \in A_0 \cup A_1, &
|\alpha|
	&= \alpha_1 + \alpha_2, &
D_z^\alpha
	&= \partial_{z_1}^{\alpha_1} \partial_{z_2}^{\alpha_2}.
\end{align}

Before proving Theorem \ref{thm:e.acc}, we require some preliminary results.  In what follows, we denote by $\Gam$ the fundamental solution corresponding to the parabolic operator $(\d_t + \Ac)$.

\begin{lemma}
For any $\delta > 0$, and $\alpha, \beta \in \Nb_0^2$ with $\beta \leq N + 2$, we have
\begin{align}
\left| (z - \zeta)^\alpha D_z^\beta \Gam(t,z;T,\zeta)\right|
	&\leq C\cdot (T-t)^{\frac{|\alpha| - |\beta|}{2}} \Gamh(t,z;T,\zeta),
&& 0 \leq t < T, \, &&z,\zeta \in I \times \Rb, \label{eq:e.full.bound}
\end{align}
and 
\begin{align}
\left| (z - \zeta)^\alpha D_\zeta^\beta \Gam(t,z;T,\zeta)\right|
	&\leq C\cdot (T-t)^{\frac{|\alpha| - |\beta|}{2}} \Gamh(t,z;T,\zeta),
&& 0 \leq t < T, \, &&z,\zeta \in I \times \Rb,  \label{eq:e.full.bound.forward}
\end{align}
where $\Gamh$ is the fundamental solution of the operator $(\d_t + (M + \del)(\d_{z_1}^2 + \d_{z_2}^2))$, and $C$ is a positive constant dependent only on $M,N,\delta$ and $|\beta|$. 
\end{lemma}

\begin{proof}
The result \eqref{eq:e.full.bound} is \cite[Lemma 6.21]{lorig-pagliarani-pascucci-4}. The inequality \eqref{eq:e.full.bound.forward} can be seen by examining the Kolmogorov forward equation. 
\end{proof}

The following fact will also be helpful. Let $a$ and $b$ be constants such that $a,b \geq 1/2$. Then, for $0 \leq t < T$
\begin{align}
\int_t^T\dd s\, (T-s)^a (s-t)^b
	&=	\frac{\Gam_E (a +1) \Gam_E (b +1) }{\Gam_E (a +b +2)}(T-t)^{a +b +1}, \label{eq:timeint}
\end{align}
where $\Gam_E$ is the Euler gamma function.

\begin{proposition}
\label{prp:derunk}
Under the assumptions of Theorem \ref{thm:e.acc}, for any multi-index $\beta \in \Nb_0^2$ we have
\begin{align}
\left|D_z^\beta u_{0,0}(t,z)\right|
	&\leq C\cdot (T -t)^{\frac{\min\{h-|\beta|,0\}}{2}}, 
&&0 \leq t < T, 
&&z, \zb \in \Rb^2.
	\label{eq:e.u00bound}
\end{align}
If $N\geq 1$, then for any $n,k \in \Nb$, $1 \leq n + k \leq N$, we have
\begin{align}
\left|D_z^\beta u_{n,k}(t,z) \right|
	&\leq C \cdot (T - t)^{\frac{n+h - |\beta|}{2}}
	\(1 + |z-\bar z|^n (T - t)^{-\frac{n}{2}}\),
&&0 \leq t < T, 
&&z,\zb \in \Rb^2.
\label{eq:e.unkbound}
\end{align}
The constants in \eqref{eq:e.u00bound} and \eqref{eq:e.unkbound} depend only on $M,N,|\beta|$ and $\| \varphi\|_{C_b^{h-1,1}}$.
\end{proposition}

\begin{proof}
The proof is analogous to the proof of \cite[Lemma 6.24]{lorig-pagliarani-pascucci-4}. 
\end{proof}

\begin{proposition}
\label{prp:e.errorPDE}
Define for $i \geq 0$
\begin{align}
\Ac_{i}^\rho
	&:= \Ac_{i,0} + \rho \Ac_{i-1,1}, 
&\Acb_n^\rho
	:= \sum_{i=0}^n  \Ac_i^\rho, \label{eq:e.Abar} \\
u_n^\rho
	&:= \sum_{i=0}^n \eps^i \rho^{n-i} u_{i,n-i}\Big|_{\eps = 1}, \label{eq:e.ubar}
\end{align}
with the convention that $\Ac_{-1,1} = 0$. Then for $N \geq 0$, we have
\begin{align}
(u - \bar u_N^\rho)(t,z)
	= \int_t^T\dd s \int_{\Rb^2} \dd \zeta\, \Gam(t,z;s,\zeta) \sum_{i=0}^N ( \Ac - \bar \Ac_i^\rho) u_{N-i}^\rho(s,\zeta). \label{eq:e.errorpde}
\end{align}
\end{proposition}
\begin{proof}
We will show that
\begin{align}
\(\d_t + \Ac \) (u -  \ub_{N}^\rho )
	+ \sum_{i=0}^{N} \(\Ac - \Acb_{i}^\rho\) u_{N-i}^\rho
	&= 0 , &
(u -  \ub_{N}^\rho )(T,\cdot)
	&=	0 ,
\label{eq:e.errorhyp}
\end{align}
from which \eqref{eq:e.errorpde} follows by an application of Duhamel's principal.  
Note that \eqref{eq:e.errorhyp} follows if we show
\begin{align}
(\d_t + \Ac) \ub_N^\rho 
	&= \sum_{i=0}^N \( \Ac - \Acb_i^\rho\) u_{N-i}^\rho , \label{eq:e.ubsum} 
\end{align}
because $(\d_t + \Ac ) u = 0$ and $u(T,\cdot) = \ub_N^\rho(T,\cdot) = \varphi$. From equations \eqref{eq:u00.pde}, \eqref{eq:unk.pde}, \eqref{eq:e.Abar} and \eqref{eq:e.ubar}, we deduce
\begin{align}
(\d_t + \Ac_{0,0})u_{n}^\rho
+ \sum_{i=1}^n \Ac_i^\rho u_{n-i}^\rho
	&= 0 . \label{eq:e.unk.pde}
\end{align}
We now proceed to show \eqref{eq:e.ubsum} by induction.  When $N = 0$, since $u_{0}^\rho  = \bar u_0^\rho$, we have
\begin{align}
\(\d_t + \Ac\) \bar u_0^\rho
	= \(\Ac - \Acb_{0}^\rho\) u_0^\rho.
\end{align}
Assume now that \eqref{eq:e.ubsum} holds for $N \geq 1$. Then we have by \eqref{eq:e.unk.pde} that
\begin{align}
\(\d_t + \Ac \)  \ub_{N+1}^\rho
	&= \(\d_t + \Ac \)  \ub_N^\rho + \( \d_t + \Ac \) u_{N+1}^\rho \\
	&= \sum_{i=0}^N \(\Ac - \Acb_i^\rho\) u_{N-i} 
		+ \( \Ac - \Acb_{0}^\rho\) u_{N+1}^\rho - \sum_{i=1}^{N+1} \Ac_i^\rho u_{N-i+1}^\rho \\
	&= \sum_{i=1}^{N+1} \(\Ac - \Acb_{i-1}^\rho\) u_{N-i+1}^\rho 
		+ \( \Ac - \Acb_{0}^\rho\) u_{N+1}^\rho - \sum_{i=1}^{N+1} \Ac_i^\rho u_{N-i+1}^\rho \\
	&=\sum_{i=0}^{N+1} \(\Ac - \Acb_{i}^\rho\) u_{N-i+1}^\rho.
\end{align}
Therefore, \eqref{eq:e.ubsum} holds for all $N$.
\end{proof}

We are now in a postition to prove Theorem \ref{thm:e.acc}.

\begin{proof}[Proof of Theorem \ref{thm:e.acc}]
From \eqref{eq:e.Abar} and \eqref{eq:e.errorpde} we have 
\begin{align}
(u - \bar u_N^\rho)(t,z)
	&= \sum_{k=0}^N \int_t^T \dd s\int_{\Rb^2} \dd \zeta\, \Gam(t,z;s,\zeta)
		\(\Ac  - \sum_{j=0}^k \(\Ac_{j,0} + \rho \Ac_{j-1,1}\)\) 
		\sum_{i=0}^{N-k} \rho^i u_{N-k-i,i}(s,\zeta). \label{eq:e.Aexpand}
\end{align}
Let $T_{k}^{a_\alpha}(z)$ be the $k$-th Taylor polynomial approximation of $a_\alpha(z)$ with the convention that $T_{-1}^{a_\alpha}(z) = 0$. We rewrite \eqref{eq:e.Aexpand} as 
\begin{align}
(u - \bar u_N^\rho)(t,z)
&= \sum_{k=0}^N \sum_{i=0}^{N-k} \rho^{i+1} \int_t^T \dd s\int_{\Rb^2} \dd \zeta\, 
		\( a_{(1,1)} - T^{a_{(1,1)}}_{k-1}\)(\zeta)\Gam(t,z;s,\zeta) D_\zeta^{(1,1)}u_{N-k-i,i}(s,\zeta) \\ & \quad
	  + \sum_{k=0}^N  \sum_{i=0}^{N-k} \sum_{ \alpha \in A_0} \rho^i
		\int_t^T \dd s\int_{\Rb^2} \dd \zeta\, 
		\( a_{\alpha} - T^{a_{\alpha}}_{k}\)(\zeta) \Gam(t,z;s,\zeta) D_\zeta^{\alpha} u_{N-k-i,i}(s,\zeta)\\
	&= \sum_{k=0}^N \sum_{i=0}^{N-k} \rho^{i+1} J_{i,k}^{(1)}
	 + \sum_{k=0}^N \sum_{i=0}^{N-k} \rho^i \( J_{i,k,1}^{(2)} + J_{i,k,2}^{(2)}\), \label{eq:e.sum1}
\end{align}
where  
\begin{align}
&J_{i,k}^{(1)}
	:= \int_t^T \dd s\int_{\Rb^2} \dd \zeta\, 
		\( a_{(1,1)} - T^{a_{(1,1)}}_{k-1}\)(\zeta) \Gam(t,z;s,\zeta) D_\zeta^{(1,1)}u_{N-k-i,i}(s,\zeta), \\
&J_{i,k,1}^{(2)}
	:= \sum_{|\alpha| \leq 1}
		\int_t^T \dd s\int_{\Rb^2} \dd \zeta\, 
		\( a_{\alpha} - T^{a_{\alpha}}_{k}\)(\zeta) \Gam(t,z;s,\zeta) D_\zeta^{\alpha} u_{N-k-i,i}(s,\zeta), \\
&J_{i,k,2}^{(2)}
	:= \sum_{\substack{ |\alpha| = 2 \\ \alpha \neq (1,1)}}
		\int_t^T \dd s\int_{\Rb^2} \dd \zeta\, 
		\( a_{\alpha} - T^{a_{\alpha}}_{k}\)(\zeta) \Gam(t,z;s,\zeta) D_\zeta^{\alpha} u_{N-k-i,i}(s,\zeta).
\end{align}
We first consider $J_{i,k}^{(1)}$. We note that $J_{i,0} = 0$ since $\Ac_{-1,1} = 0$ by convention. For $k \geq 1$, we perform integration by parts to obtain for $|\alpha_1| = |\alpha_2| = 1$,
\begin{align}
J_{i,k}^{(1)} 
	&= -\int_t^T \dd s \int_{\Rb^2} \dd \zeta\, 
		\[D_\zeta^{\alpha_1} \( a_{(1,1)} - T_{k-1}^{a_{(1,1)}}\)(\zeta) \Gam(t,z;s,\zeta)\]
		\[D_\zeta^{\alpha_2} u_{N-k-i,i}(s,\zeta)\]. 
\end{align}
By the product rule and \eqref{eq:e.full.bound.forward}, evaluating at $z = \zb$ gives
\begin{align}
|J_{i,k}^{(1)}|
	&\leq C_1 \int_t^T \dd s \int_{\Rb^2} \dd \zeta\, 
		|z - \zeta|^{k-1} \Gam(t,z;s,\zeta) \left| D_\zeta^{\alpha_2}	u_{N-k-i,i}(s,\zeta)\right| \\ &\quad
	    + C_2 \int_t^T \dd s \int_{\Rb^2} \dd \zeta\, 
	        |z - \zeta|^{k} \left| D_z^{\alpha_1} \Gam(t,z;s,\zeta)\right|\, \left| D_\zeta^{\alpha_2} u_{N-k-i,i}(s,\zeta)\right|.
\end{align}
Applying \eqref{eq:e.full.bound} and \eqref{eq:e.unkbound} gives 
\begin{align}
|J_{i,k}^{(1)}|
	&\leq C_3 \int_t^T \dd s\int_{\Rb^2}\dd \zeta\, \Gamh(t,z;s,\zeta) (s-t)^{\frac{k-1}{2}} (T-s)^{\frac{N+h-k-i-1}{2}}
		\(1+|z - \zeta|^{N-k-i} (T-s)^{\frac{N-k-i}{2}}\) \\
	&\leq C_4 \int_t^T\dd s \[ (s-t)^{\frac{k-1}{2}}(T-s)^{\frac{N+h-k-i-1}{2} } +(s-t)^{\frac{N-i-1}{2}}(T-s)^{\frac{h-1}{2}}\] 
	\int_{\Rb^2}\dd \zeta\, \Gamh_2(t,z;s,\zeta) \\
	&\leq C_5\, (T-t)^{\frac{N+h-i}{2}}.
		&&\text{(by } \eqref{eq:timeint}\text{)}
\end{align}
Similar arguments show
\begin{align}
&|J_{i,k,1}^{(2)}|
	\leq C_6\, (T-t)^{\frac{N+h-i+2}{2}},
&|J_{i,k,2}^{(2)}|
	\leq C_7\, (T-t)^{\frac{N+h-i+1}{2}},
\end{align}
for $0 \leq k \leq N$. When $N = 0$, $J_{i,k}^{(1)} = 0$, so by \eqref{eq:e.sum1} we have
\begin{align}
|(u-\ub_0^\rho)(t,z)|
	\leq C_8\, (T-t)^{\frac{h+1}{2}}.
\end{align}
When $N\geq 1$, by \eqref{eq:e.sum1} we have
\begin{align}
|(u-\bar u_N^\rho)(t,z)|
	&\leq C_9\, \sum_{k=0}^N \sum_{i=0}^{N-k} \left( |\rho|^{i+1} (T-t)^{\frac{N-i+h}{2}} + |\rho|^{i}(T-t)^{\frac{N+h-i+1}{2}} \right)\\
	&\leq C_{10}\, ( (T-t)^{\frac 12} + |\rho|)  \sum_{k=0}^N \sum_{i=0}^{N-k} |\rho|^{i} (T-t)^{\frac{N-i+h}{2}}\\
	&= C_{11}\, ( (T-t)^{\frac 12}+ |\rho|) \sum_{k=0}^N (N-k+1) |\rho|^k (T-t)^{\frac{N-k+h}{2}} \\
	&\leq C_{12}\, ((T-t)^{\frac 12} + |\rho|)  \sum_{k=0}^N  |\rho|^k (T-t)^{\frac{N-k+h}{2}} ,
\end{align}
which proves Theorem \ref{thm:e.acc}.
\end{proof}

%
%

\bibliographystyle{chicago}
\bibliography{barrier-bib}	

%
%

\clearpage

\begin{figure}[ht]

	\begin{minipage}[t]{.45\textwidth}
		\includegraphics[width=\textwidth]{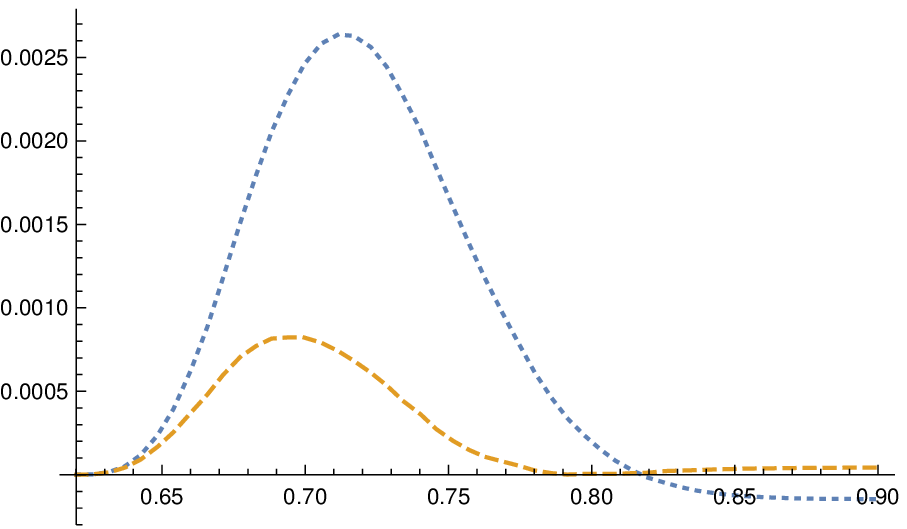}
		\caption{
		  For the Heston model considered in Section \ref{sec:heston}, we plot $u-\bar u_0^\rho$ (blue dotted) and $u-\bar u_2^\rho $ (orange dotted-dashed) as a function of {the upper barrier $U$} for a call option.
		}
		\label{fig:ATMcallerror}
	\end{minipage}
	\hfill
	\begin{minipage}[t]{.45\textwidth}
		\includegraphics[width=\textwidth]{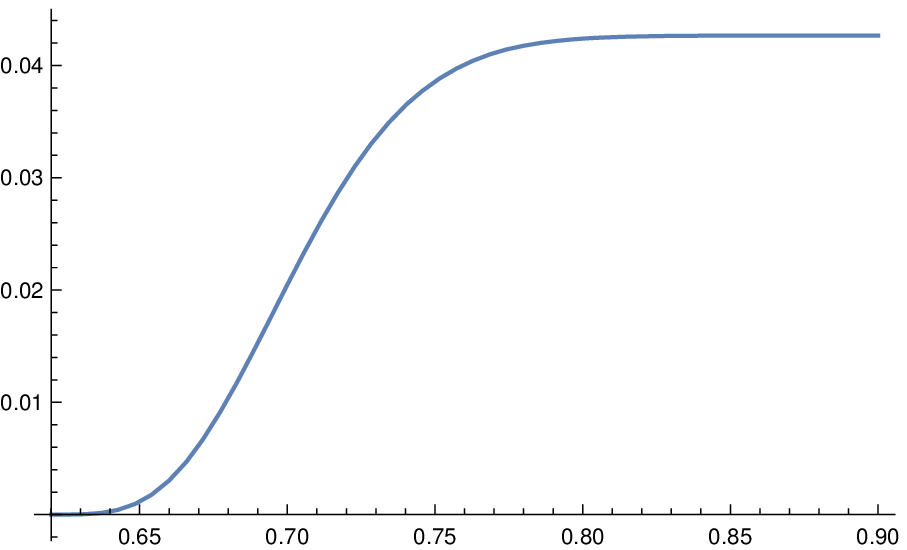}
		\caption{
			For the Heston model considered in Section \ref{sec:heston}, we plot $u$ as a function of {the upper barrier $U$} for a call option.
		}
		\label{fig:ATMcall}
	\end{minipage}

\end{figure}

\begin{figure}[ht]

		\begin{minipage}[t]{.45\textwidth}
		\includegraphics[width=\textwidth]{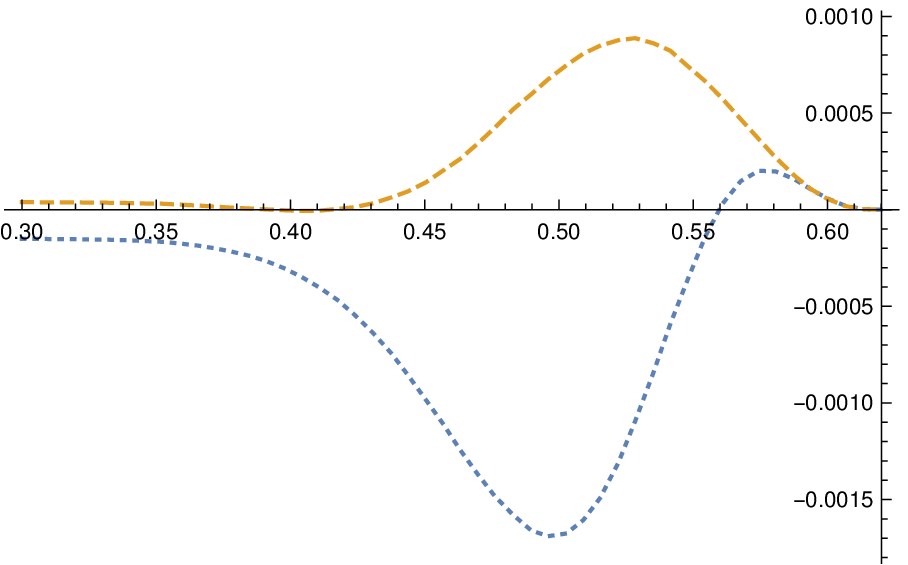}
		\caption{
			For the Heston model considered in Section \ref{sec:heston}, we plot $u - \ub_0^\rho$ (blue dotted) and $u-\ub_2^\rho$ (orange dotted-dashed) as a function of {the lower barrier $L$} for a put option.
		}
		\label{fig:ATMputerror}
	\end{minipage}
	\hfill
	\begin{minipage}[t]{.45\textwidth}
		\includegraphics[width=\textwidth]{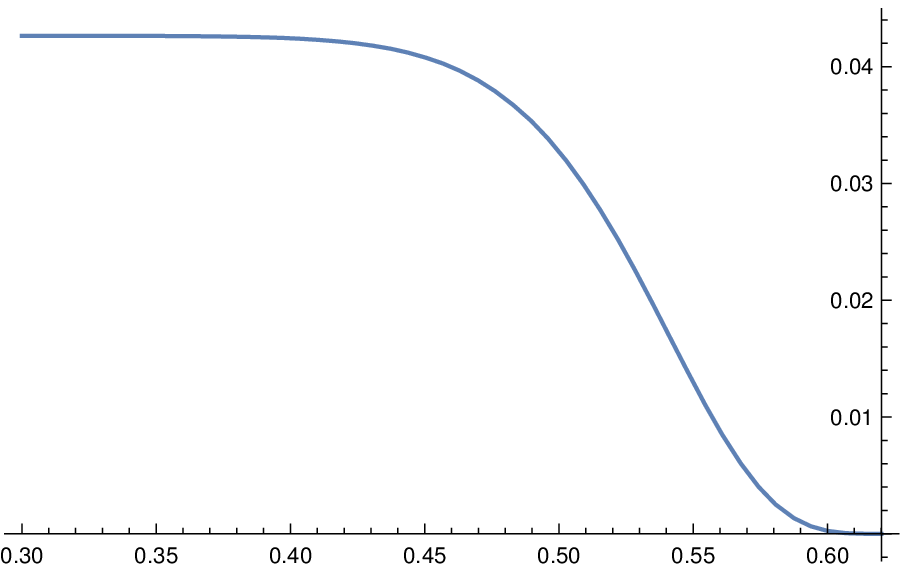}
		\caption{
			For the Heston model considered in Section \ref{sec:heston}, we plot $u$ as a function of {the lower barrier $L$} for a put option.
		}
		\label{fig:ATMput}
	\end{minipage}
 
\end{figure}

\begin{figure}[ht]

	\begin{minipage}[t]{.45\textwidth}
		\includegraphics[width=\textwidth]{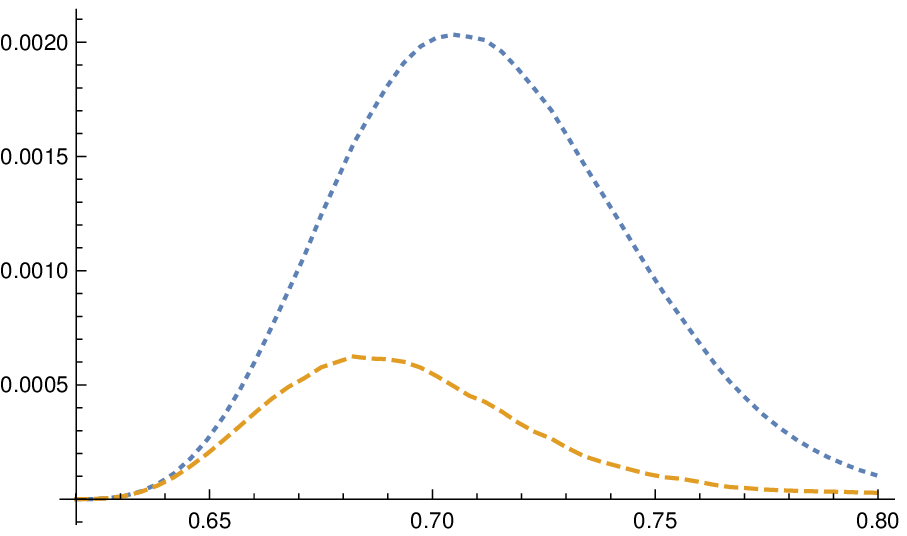}
		\caption{
			For the CEV model considered in Section \ref{sec:cev}, we plot $u-\ub_0$ (blue dotted) and $u-\ub_2 $ (orange dashed) as a function of {the upper barrier $U$} for a call option.
		}
		\label{fig:ATMCEVcallerror}
	\end{minipage}
	\hfill
	\begin{minipage}[t]{.45\textwidth}
		\includegraphics[width=\textwidth]{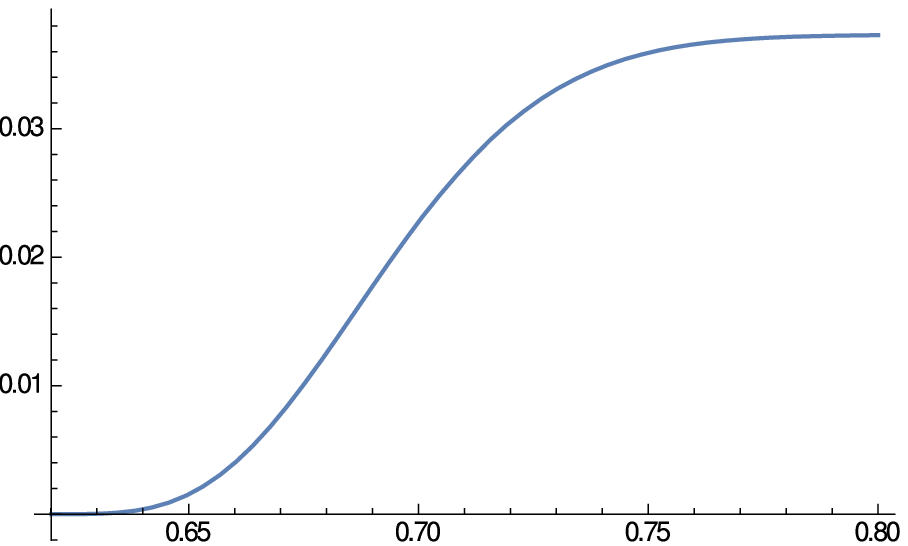}
		\caption{
			For the CEV model considered in Section \ref{sec:cev}, we plot $u$ as a function of {the upper barrier $U$} for a call option.
		}
		\label{fig:ATMCEVcall}
	\end{minipage}
\end{figure}

\begin{figure}[ht]

	\begin{minipage}[t]{.45\textwidth}
		\includegraphics[width=\textwidth]{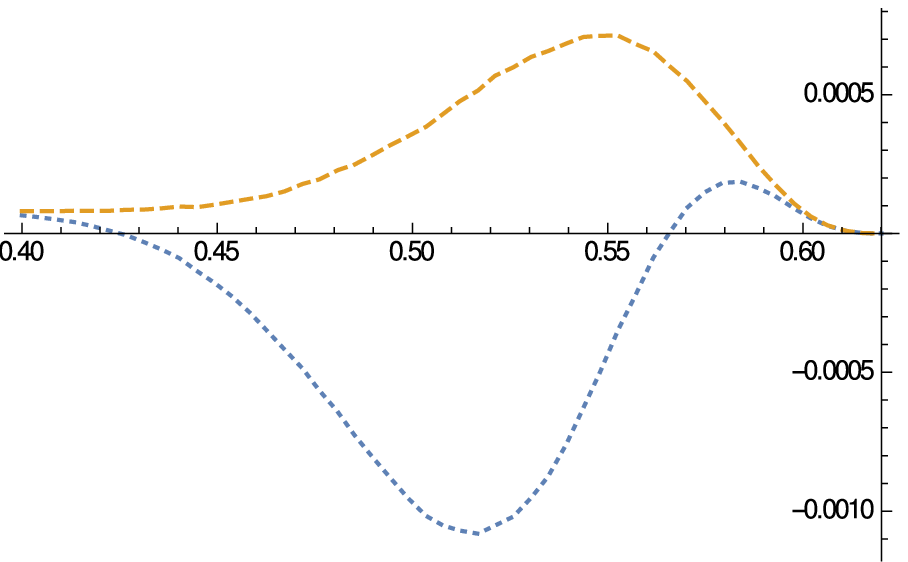}
		\caption{
			For the CEV model considered in Section \ref{sec:cev}, we plot $u-\bar u_0$ (blue dotted) and $u-\bar u_2 $ (orange dashed) as a function of {the lower barrier $L$} for a put option.
		}
		\label{fig:ATMCEVputerror}
	\end{minipage}
	\hfill
	\begin{minipage}[t]{.45\textwidth}
		\includegraphics[width=\textwidth]{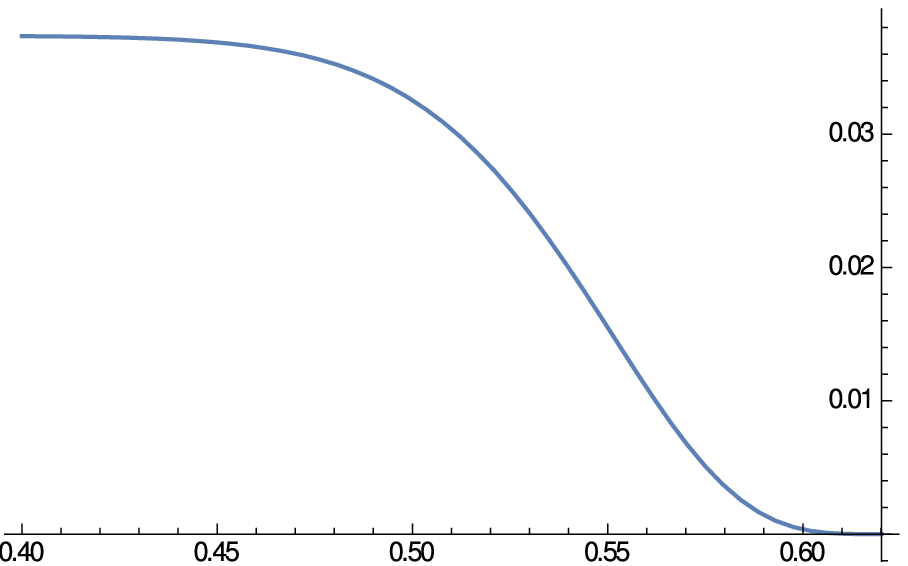}
		\caption{
			For the CEV model considered in Section \ref{sec:cev}, we plot $u$ as a function of {the lower barrier $L$} for a put option.
		}
		\label{fig:ATMCEVput}
	\end{minipage}
\end{figure}

\end{document}